\documentclass[journal,doublecolumn]{IEEEtran}
\usepackage{mathrsfs}
\usepackage{amsfonts}
\usepackage{epsfig,latexsym}
\usepackage{float}
\usepackage{indentfirst}
\usepackage{amsmath}
\usepackage{amssymb}
\usepackage{times}
\usepackage{subfigure}
\usepackage{psfrag}
\usepackage{cite}
\usepackage{algorithmic,algorithm}

\newtheorem{theorem}{$\mathbf{Theorem}$}

\newtheorem{proposition}[theorem]{Proposition}
\newtheorem{corollary}[theorem]{$\mathbf{Corollary}$}

\begin{document}
\title{Cluster Content Caching: An Energy-Efficient Approach to Improve Quality of Service in Cloud Radio Access Networks
}

\author{Zhongyuan Zhao, \IEEEmembership{Member, IEEE}, Mugen Peng, \IEEEmembership{Senior Member, IEEE}, Zhiguo Ding, \IEEEmembership{Senior Member, IEEE}, Wenbo Wang, \IEEEmembership{Senior Member, IEEE}, and H. Vincent Poor, \IEEEmembership{Fellow, IEEE}
\thanks{Zhongyuan~Zhao (e-mail: zyzhao@bupt.edu.cn), Mugen~Peng (e-mail: pmg@bupt.edu.cn), and Wenbo Wang (e-mail: wbwang@bupt.edu.cn) are with the Key Laboratory of Universal Wireless Communications (Ministry of Education), Beijing University of Posts and Telecommunications, Beijing, China. Zhiguo Ding (z.ding@lancaster.ac.uk) is with the School of Computing and Communications, Lancaster University, Lancaster, LA1 4YW, UK. H. Vincent Poor (poor@princeton.edu) is with the Department of Electrical Engineering, Princeton University, Princeton, NJ, USA.}
\thanks{The work of Z. Zhao, M. Peng and W. Wang was supported by National Natural Science Foundation of China (Grant No. 61501045, No. 61361166005), the  National High Technology Research and Development Program of China (Grant No. 2014AA01A701), the State Major Science and Technology Special Projects (Grant No. 2016ZX03001020-006), and the Fundamental Research Funds for the Central Universities. The work of Z. Ding was supported by the UK EPSRC under grant number EP/L025272/1 and by H2020-MSCA-RISE-2015 under grant number 690750. The work of H. V. Poor was supported in part by the U.S. National Science Foundation under Grant ECCS-1343210.}}
\maketitle
\begin{abstract}
In cloud radio access networks (C-RANs), a substantial amount of data must be exchanged in both backhaul and fronthaul links, which causes high power consumption and poor quality of service (QoS) experience for real-time services. To solve this problem, a cluster content caching structure is proposed in this paper, which takes full advantage of distributed caching and centralized signal processing. In particular, redundant traffic on the backhaul can be reduced because the cluster content cache provides a part of required content objects for remote radio heads (RRHs) connected to a common edge cloud. Tractable expressions for both effective capacity and energy efficiency performance are derived, which show that the proposed structure can improve QoS guarantees with a lower power cost of local storage. Furthermore, to fully explore the potential of the proposed cluster content caching structure, the joint design of resource allocation and RRH association is optimized, and two distributed algorithms are accordingly proposed. Simulation results verify the accuracy of the analytical results and show the performance gains achieved by cluster content caching in C-RANs.
\end{abstract}


\IEEEpeerreviewmaketitle

\vspace{0.25in}
\section{Introduction}
The cloud radio access network (C-RAN) is an energy-efficient network architecture for mobile operators to provide high data rate service, in which remote radio heads (RRHs) connect with a cloud-based baseband unit (BBU) via fiber fronthaul links \cite{b1}. In addition to savings on capital and operational expenditures and reducing power consumption, the centralization of BBUs enables large-scale signal processing and resource management, and the theoretical performance limits have been studied in \cite{b2}. In practical transmission situations, the capacity of both the fronthaul and the backhaul may not be able to support the significant number of low-latency data exchanges between RRHs and the BBU pool. Although some efficient signal compression methods are proposed in C-RANs, such as \cite{b11} and \cite{b14}, it is insufficient to satisfy the dramatically increasing requirements from mobile users for real-time services with high quality of service (QoS) guarantees. To solve this problem, some evolved network architectures of C-RANs have been proposed, such as heterogeneous cloud radio access networks (H-CRANs) in \cite{b80} and \cite{b6} and edge cloud radio access networks (EC-RANs) in \cite{b4}. In particular, H-CRANs and EC-RANs can balance the loadings of the fronthaul and the backhaul by separating the controller from the cloud center and decentralizing the BBU pools, respectively \cite{b81}. The performance gains of H-CRANs and EC-RANs have been evaluated in \cite{b10} and \cite{b9}, which demonstrated that they have great potentials to improve both spectral and energy efficiencies.

A key network architecture feature of C-RANs and their related extensions is that the distance between a user and its accessed RRH is significantly shortened, which suggests that wireless networks are evolving from a base station-centric architecture to a user-centric architecture, as well as from a connection-centric purpose to a content-centric purpose. Compared with the conventional network architectures, content-centric networks pay more attention to QoS guarantees, and some dynamic storage units, termed content caches are employed. In \cite{b61} and \cite{b62}, the concept of proactive caching has been proposed in heterogenous networks (HetNets), where the content caches are deployed at the edge of networks, such as the small cells and the users. The performance improvement achieved by proactive caching has been verified in \cite{b63}. However, these edge caching strategies are not applicable in C-RANs due to the absence of the BBU at RRHs. In particular, the content in edge caches should be transferred to the cloud BBU pool, and then the radio frequency version of information will be sent back to RRHs. Therefore, the content delivery with edge caches cannot be accomplished between RRHs and the users only, which causes poor delay experience. In the paradigm of C-RANs, content caches are equipped at the cloud center to fully exploit the computational capability. In \cite{b20}, a centralized content caching structure was proposed, and a test content-centric network based on the architecture of C-RANs was accordingly introduced. A hierarchical content caching structure has been proposed in \cite{b18}, and the coordination between the cloud cache and the edge caches at the base stations was optimized in \cite{b19}.

Although content caching is valid to improve the QoS and the energy efficiency of wireless networks, it is hard to find a suitable metric to evaluate the performance gains precisely. The existing works, such as in \cite{b2}, cannot characterize the impacts of content caching on the QoS in C-RANs. Based on the constant source arrival rate assumption, effective capacity has been defined as a link-level QoS metric for a wireless channel in \cite{b22}. The concept of effective capacity has been shown to be an efficient criterion to capture the delay experience of data services, which has been studied in the basic scenarios of wireless communications \cite{b24,b40,b25}.

{Content-centric C-RANs have substantial potential with several challenges: First, although the existing content-centric C-RAN concept can reduce the service delay, it is not easy to decide where to deploy the content caches since it is a dilemma to balance the centralized signal processing and the distributed caching. Secondly, the existing studies of QoS evaluation in wireless systems mostly focus on interference-free scenarios, which are not suitable for analyzing the performance of content caching in the typical interference-limited C-RANs. Thirdly, due to the limited storage volume, it is still not straightforward to fully exploit the potential of content caches in C-RANs.}

Therefore, motivated by the necessity of network architecture enhancement and QoS metric establishment, a cluster content caching structure in C-RANs is studied in this paper, and the performance analysis and optimization algorithm design are studied as well. Our main contributions can be summarized as follows.

{\textit{Firstly}, a cluster content caching structure is proposed in edge cloud-radio access networks (EC-RANs), in which RRHs in the same cluster share a common local cluster content cache. Our proposed structure can take full advantages of centralized signal processing and distributed caching in C-RANs, and the capacity constraint of backhaul links is analyzed to show the improvement of delay experience achieved by our proposed structure.

\textit{Secondly}, by formulating a stochastic geometry-based network model, both the effective capacity and the energy efficiency of the proposed cluster content caching structure are analyzed, and the corresponding tractable expressions are derived. The analytical results show that our proposed structure can improve the QoS guarantees in an energy-efficient way because some requests are locally responded to with a low cost of storage in each cluster. Simulation results show that the effective capacity and the energy efficiency can be improved up to 0.57 Mbit/s/Hz and 0.004 Mbit/Joule when the number of required content objects is five.

\textit{Thirdly}, to further improve the performance gains of cluster content caching in C-RANs, the joint design of radio resource unit (RRU) allocation and RRH association is optimized, which can be solved by a nested coalition formation game. Moreover, to reduce the computational complexity, the RRU allocation problem is solved by applying the Shapley value. By employing the proposed optimization algorithms, the performance gains can be enlarged to 0.95 Mbit/s/Hz and 0.0055 Mbit/Joule, respectively.}

The rest of this paper is organized as follows. Section \uppercase\expandafter{\romannumeral2} describes the system model, and the cluster content caching structure in C-RANs is also proposed, which shows that the proposed structure can reduce the loadings on backhaul links. In Section \uppercase\expandafter{\romannumeral3}, both the effective capacity and the energy efficiency are studied, and the tractable expressions will be provided. In Section \uppercase\expandafter{\romannumeral4}, the joint design of RRU allocation and RRH association is studied, and a nested coalition formation game-based algorithm is provided. To reduce the computational complexity, a sub-optimal algorithm is given in Section \uppercase\expandafter{\romannumeral5}. The simulation results are shown in Section \uppercase\expandafter{\romannumeral6}, followed by the conclusion in Section \uppercase\expandafter{\romannumeral7}.

\section{System Model and Protocol}

To improve the QoS guarantees of real time services, content caches are employed in C-RANs. For conventional C-RANs as illustrated in Fig. \ref{fig:subfig:fig1a}, a cloud based content cache is included at the cloud center, where both a cloud BBU pool and a centralized controller are included as well. RRHs connect to the BBU pool via fronthaul, while the BBU pool connects to the content cloud through backhaul. In this fully centralized network architecture, heavy burdens are imposed on both backhaul and fronthaul due to a substantial amount of data exchanging between the cloud center and RRHs, which causes high power consumption and long delay. To solve this problem, a paradigm of EC-RANs with cluster content caching is studied in this paper. As shown in Fig. \ref{fig:subfig:fig1b}, the loadings on both fronthaul and backhaul links can be balanced by partially decentralizing the baseband signal processing and the control functions into the edge devices, such as base stations and user equipments with cache. Moreover, by employing a cluster content cache in each cluster, the QoS guarantee and energy efficiency can be further improved because some requests are immediately responded to in the cluster with a low cost on local storage. Note that the cluster is dynamically formed in a edge manner according to the transmitting packet services, the wireless transmit capabilities, and etc.

Without loss of generality, we focus on a typical cluster $\mathcal{C}_T$ in Fig. \ref{fig:subfig:fig1b}, where RRHs are connected to a common edge cloud, and the cluster-scale joint management, such as scheduling and resource allocations, can be implemented. Both a cluster content cache $\mathcal{U}_T$ and a cloud content cache $\mathcal{U}_C$ are deployed to fully exploit the potential of content caching in C-RANs. The content requests from served users are aggregated at the edge cloud in $\mathcal{C}_T$, and can be treated by using the following strategy: First, $\mathcal{C}_T$ checks its cluster content cache $\mathcal{U}_T$, and the requests can be served immediately if the desired content is available at $\mathcal{U}_T$. Otherwise, the requests will be forwarded to the cloud content cache, and then the corresponding content can be provided through the backhaul link from the content cloud. Then the requests can be handled similarly to the case in which the content resides in the cluster content cache $\mathcal{U}_T$.
\begin{figure} \centering
\subfigure{\label{fig:subfig:fig1a} \includegraphics[width=2.5in]{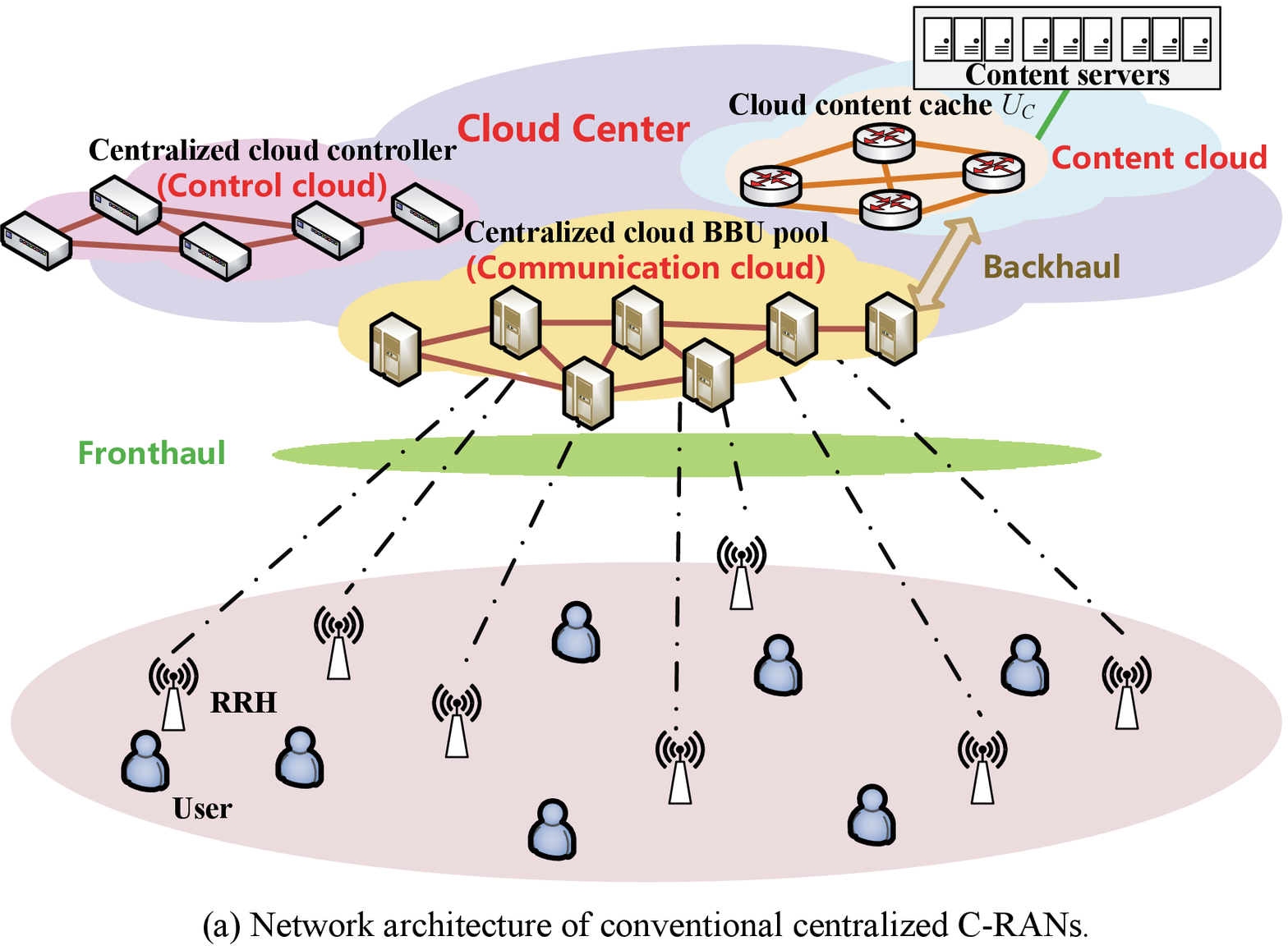}}
\subfigure{\label{fig:subfig:fig1b} \includegraphics[width=2.5in]{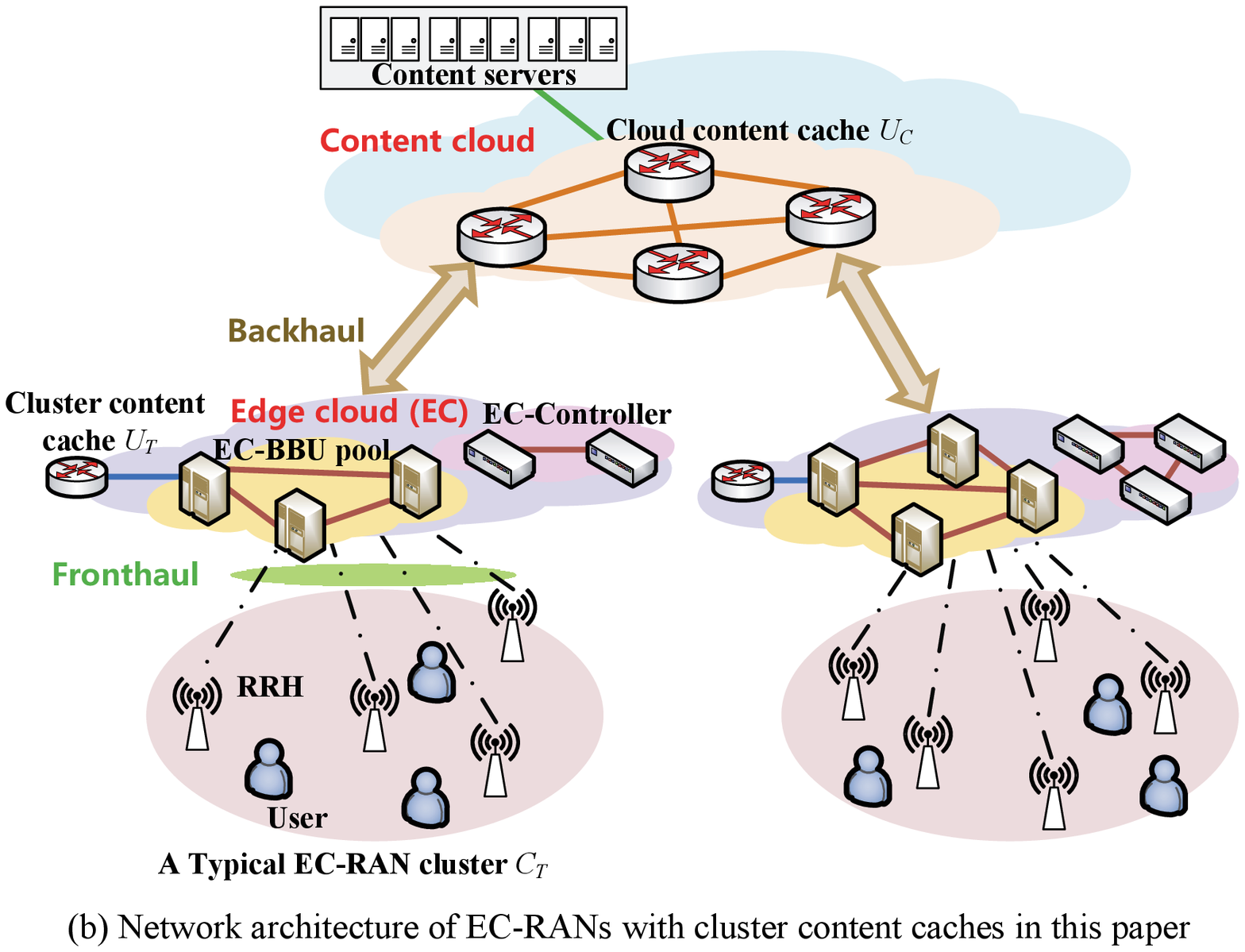}}
\caption{The illustrations of C-RANs and EC-RANs with content caches.} \label{fig:subfig}
\end{figure}

{For convenience, we assume that the cloud content cache $\mathcal{U}_C$ stores the set of all the content objects required by the users in $\mathcal{C}_T$, which is denoted as a limited set $\Omega_C = \{S_1,\ldots,S_L\}$, with $S_1,\ldots,S_L$ all of the same size $B_S$. Moreover, the content objects kept in the local cluster cache $\mathcal{U}_T$, which are selected randomly from $\Omega_C$, can be treated as the members belonging to a subset of $\Omega_C$, e.g., $\Omega_T = \{S_{T_1},\ldots,S_{T_K}\} \subseteq \Omega_C$. A radio resource unit (RRU) is defined as a set of resources in the time and frequency domains that are allocated to accomplish the transmission of a specific content object through radio access channels. Each RRH can serve only one content object in each RRU due to the single-antenna deployment.} To provide a tractable expression of channel capacity for a link from RRH to user in C-RANs, the observation at a typical user $\mathrm{u}_T$ in $\mathcal{C}_T$ can be expressed as
{\begin{eqnarray}\label{eqn:obs}
y_T = \sqrt{\rho} h_m d_m^{-\beta / 2} s_m + \sum_{\mathrm{R}_j \in \mathcal{I}_\mathrm{R}} \sqrt{\rho} h_j d_j^{-\beta / 2} s_j + n_T,
\end{eqnarray}}
where $s_m$ is the desired message for $\mathrm{u}_T$ from its serving RRH $\mathrm{R}_m$, the channels are modeled by including Rayleigh fading and path loss in this paper, e.g., $h_m$ denotes the flat Rayleigh fading for the link from $\mathrm{R}_m$ to $\mathrm{u}_T$, $d_m$ is the distance between $\mathrm{R}_m$ and $\mathrm{u}_T$, $\beta$ is the path loss exponent, {$\mathcal{I}_\mathrm{R}$ denotes the set of all the interfering RRHs}, $s_j$, $h_j$ and $d_j$ are defined similarly for a specific interfering RRH $\mathrm{R}_j$, $j \neq m$, $h_m, h_j \sim \mathcal{CN}(0,1)$, $n_T$ is the additive Gaussian noise with unit covariance, and $\rho$ is the average signal-to-noise ratio (SNR). Due to \eqref{eqn:obs}, the channel capacity can be expressed as follows:
\begin{equation}\label{eqn:shannon}
C = \mu W \log(1 + \gamma),~\mathrm{where}~\gamma = \frac{\rho d_m^{-\beta} |h_m|^2}{\sum_{\mathrm{R}_j \in \mathcal{I}_\mathrm{R}} \rho d_j^{-\beta} |h_j|^2 + \sigma^2},
\end{equation}
where $\mu$ denotes the spectral efficiency. {In particular, $\mu$ determines how many content objects can be served in each RRU, e.g., $\mu = L B_S / (NWT)$ bit/s/Hz when $L$ content objects can be served in $N$ RRUs, where the time length and the spectrum bandwidth of a RRU are denoted by $T$ and $W$, respectively.

As shown in \eqref{eqn:shannon}, the channel capacity is based on the ideal delay assumption, and thus cannot characterize the QoS of requested content. To provide a tractable QoS metric of a wireless channel, the concept of effective capacity is introduced as follows.

\subsection{The Theory of Effective Capacity in Wireless Channels}
It is challenging to evaluate QoS that can be supported by a wireless link due to its inconstant transmission condition. In \cite{b22}, the concept of effective capacity is provided as a feasible solution of this problem. In particular, assuming there exists a queue of infinite size with a constant arrival rate at the source, effective capacity can characterize the maximum arrival rate that can be supported by a wireless channel with a specific QoS guarantee, which is defined as follows \cite{b22}:
\begin{eqnarray}\label{eqn:effect1}
E(\theta) = -\lim \limits_{t \rightarrow \infty} \frac{1}{\theta t} \log \mathbb{E} \big \{e^{- \theta S(t)}\big\},
\end{eqnarray}
where $S(t) = \sum_{0 = t_0<t_1< \cdots < t_n = t} \int_{t_{i-1}}^{t_i} r(\tau) \mathrm{d} \tau$ is the delivered service through a wireless channel in bits over the time interval [0, t), $r(\tau)$ denotes the channel capacity at time $\tau$, and $\theta$ is the QoS exponent. To specify the delay constraint, $\theta$ is defined as the decay rate of the tail distribution with a stochastic queue length $Q$,
\begin{eqnarray}\label{eqn:theta}
\theta = \lim \limits_{q \rightarrow \infty} \frac{\log \Pr\{Q > q\}}{q}.
\end{eqnarray}
For a large value of threshold $q_{\max}$, the buffer violation probability can be approximated as $\Pr\{Q > q_{\max}\} \approx e^{- \theta q_{\max}}$, and the delay violation probability can be bounded as $\Pr\{D > d_{\max}\} \leqslant c \sqrt{\Pr\{Q > q_{\max}\}}$, where $c$ is a positive constant related to the arrival rate \cite{b22}. Therefore, a smaller $\theta$ implies a looser QoS requirement, while larger $\theta$ denotes a stricter QoS requirement. Due to \eqref{eqn:effect1} and \eqref{eqn:theta}, the effective capacity can characterize the relationship between delay experience and wireless channel capacity.

Under the block fading channel assumption, each channel coefficient is a constant in each RRU, and the effective capacity can be further derived as \cite{b40}
\begin{eqnarray}\label{eqn:effect2}
E(\theta) = - \frac{1}{\theta W \bar{T}} \ln \mathbb{E} \big \{(1 + \gamma)^{- \mu \theta W \bar{T}}\big\},
\end{eqnarray}
where $\bar{T} = T / \ln2$.}

\begin{figure} \centering
\subfigure{\label{fig:subfig:quemod_a} \includegraphics[width=3in]{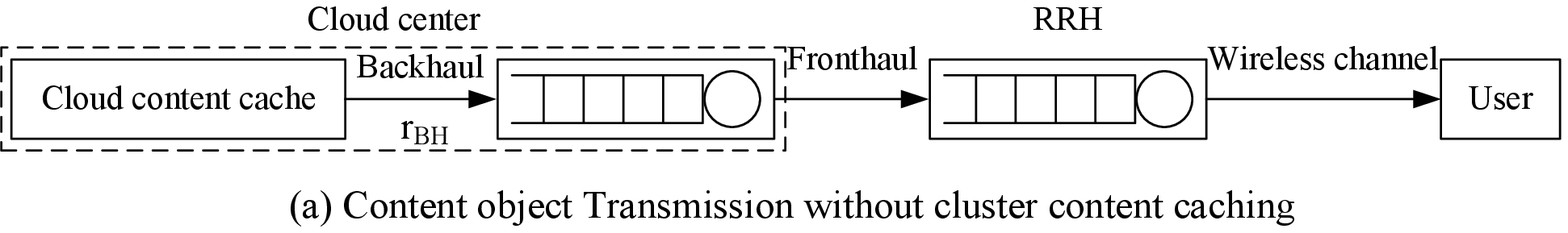}}
\subfigure{\label{fig:subfig:quemod_b} \includegraphics[width=3in]{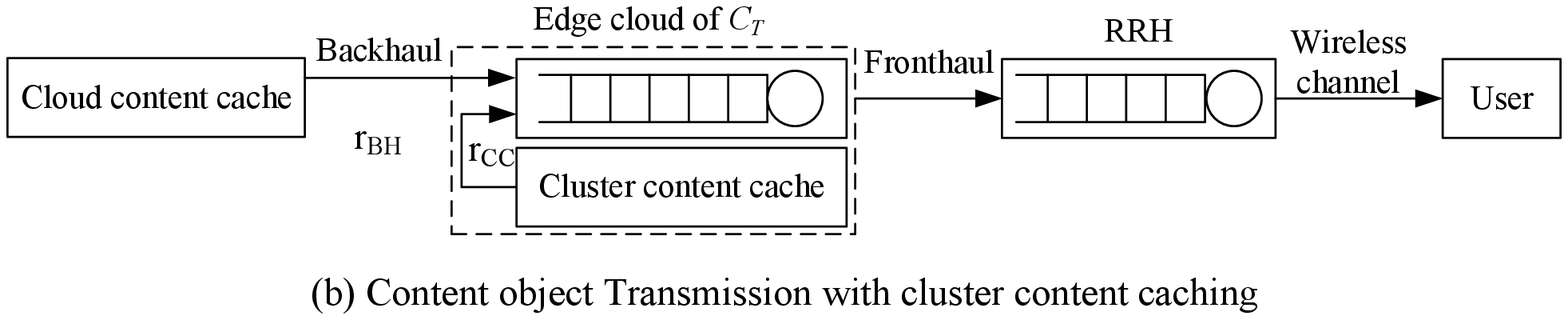}}
\caption{Queueing model of content transmissions in C-RANs.
} \label{fig:que mod}
\end{figure}

\subsection{Capacity Constraint of Backhaul in C-RANs with Cluster Content Caching}
Since some objects can be obtained locally, cluster content caching improves the QoS guarantees by reducing delay, and migrates the loading on backhaul links. As shown in Fig. \ref{fig:que mod}, our studied content objects are sent to users via multi-hop tandem links. In particular, a content object can be transmitted to the edge cloud in each C-RAN cluster through a wired backhaul link or from a local cluster content cache, whose arrival data rates are fixed and can be denoted as $r_{\mathrm{BH}}$ and $r_{\mathrm{CC}}$, respectively. Then it is delivered to RRHs through fiber fronthaul links, and finally can be sent to the users by using wireless channels. The corresponding experienced delay can be given as the following proposition.
\begin{proposition}
(Proposition 5, \cite{b25}) Assume that a network carries packetized traffic, and consists of $N_{\mathrm{h}}$ hops. Given an external arrival process with constant data rate $r$ and constant packet size $B$, the end-to-end delay $D$ experienced by the traffic can be expressed as
\begin{eqnarray}\label{eqn:del qos}
\lim \limits_{D_{\max} \rightarrow \infty} \frac{\log \Pr \{D > D_{\max}\}}{D_{\max} - (N_{\mathrm{h}} B) / r} = - \theta,
\end{eqnarray}
where $\theta$ denotes the QoS exponent defined by \eqref{eqn:theta}, and the effective capacity satisfies $E(\theta / r) = r$. Based on \eqref{eqn:del qos}, the following approximation can be obtained for a large delay threshold $D_{\max}$:
\begin{eqnarray}\label{eqn:del ineq}
\Pr \{D > D_{\max}\} \approx e^{- \theta (D_{\max} - \frac{N_{\mathrm{h}} B}{r})}.
\end{eqnarray}
\end{proposition}

Proposition 1 implies that a higher arrival data rate $r$ leads to a smaller QoS exponent $\theta$, and improves delay experience. The QoS exponents of a content object $S_l$ that is responded by the cluster content cache $\mathcal{U}_T$ and the cloud cache $\mathcal{U}_C$ are denoted as $\theta^{\mathrm{T}}_l$ and $\theta^{\mathrm{C}}_l$, respectively. Based on Proposition 1, the following theorem of the relationship between $\theta^{\mathrm{T}}_l$ and $\theta^{\mathrm{C}}_l$ can be provided.
\begin{theorem}
To achieve the same delay experience, the QoS exponents of content cloud and cluster caching to deliver a specific content object $S_l$ should satisfy the following condition:
\begin{eqnarray}\label{eqn:qos cc lc}
\theta^{\mathrm{C}}_l = \frac{1}{1 - \frac{2 B_S}{r_{\mathrm{BH}} D_{\max}}} \theta^{\mathrm{T}}_l.
\end{eqnarray}
\end{theorem}
\begin{proof}
Based on \eqref{eqn:del qos}, the QoS exponent of cloud content caching can be expressed as
\begin{eqnarray}\label{eqn:qos exp cc}
\theta^{\mathrm{C}}_l = - \lim \limits_{D_{\max} \rightarrow \infty} \frac{\log \Pr \{D > D_{\max}\}}{D_{\max} - (2 B_S) / r_{\mathrm{BH}}}.
\end{eqnarray}
Similarly, the QoS exponent of cluster content caching can be given as
\begin{eqnarray}\label{eqn:qos exp lc}
\nonumber \theta^{\mathrm{T}}_l &=& - \lim \limits_{D_{\max} \rightarrow \infty} \frac{\log \Pr \{D > D_{\max}\}}{D_{\max} - (2 B_S) / r_{\mathrm{CC}}} \\
&\overset{\mathrm{(a)}}{\approx}& - \lim \limits_{D_{\max} \rightarrow \infty} \frac{\log \Pr \{D > D_{\max}\}}{D_{\max}},
\end{eqnarray}
where the approximation (a) in \eqref{eqn:qos exp lc} follows the fact that the arrival rate $r_{\mathrm{CC}}$ goes infinity since the content can be obtained locally in the C-RAN cluster, and the delay caused by arrival process can be ignored. When both the cloud and the cluster caching structures share a common delay experience, $\theta^{\mathrm{C}}_l$ in \eqref{eqn:qos exp cc} can be further derived as
\begin{eqnarray}\label{eqn:qos exp cc1}
\nonumber \frac{1}{\theta^{\mathrm{C}}_l} &=& - \lim \limits_{D_{\max} \rightarrow \infty} \frac{D_{\max} - (2 B_S) / r_{\mathrm{BH}}}{\log \Pr \{D > D_{\max}\}} \\
&=& \frac{1}{\theta^{\mathrm{T}}_l} + \lim \limits_{D_{\max} \rightarrow \infty} \frac{(2 B_S) / r_{\mathrm{BH}}}{\log \Pr \{D > D_{\max}\}}.
\end{eqnarray}
Recalling \eqref{eqn:del ineq} in Proposition 1, the probability of delay $D$ exceeding a threshold $D_{\max}$ can be given as $\lim_{D_{\max} \rightarrow \infty} \Pr \{D > D_{\max}\} \approx e^{- \theta^{\mathrm{T}}_l D_{\max}} $. Then \eqref{eqn:qos exp cc1} can be given as
\begin{eqnarray}
\frac{1}{\theta^{\mathrm{C}}_l} = \frac{1}{\theta^{\mathrm{T}}_l} + \frac{(2 B_S) / r_{\mathrm{BH}}}{\log e^{- \theta^{\mathrm{T}}_l D_{\max}}} = \bigg(1 - \frac{2 B_S}{r_{\mathrm{BH}} D_{\max}}\bigg)\frac{1}{\theta^{\mathrm{T}}_l}.
\end{eqnarray}
And the proof has been finished.
\end{proof}

Theorem 2 shows that the cluster content caching can reduce the loading on the backhaul links. In particular, due to \eqref{eqn:qos cc lc}, the capacity of backhaul links must satisfy the following constraint:
\begin{eqnarray}\label{eqn:bh cons}
c_{\mathrm{BH}} \geqslant r_{\mathrm{BH}} = \frac{2B_S}{D_{\max} (1 - \frac{\theta^{\mathrm{T}}_l}{\theta^{\mathrm{C}}_l})}.
\end{eqnarray}
This implies that the only way to reduce the performance gap of QoS guarantee with and without cluster content caching is to enlarge the capacity of backhaul link. In particular, as the QoS guarantee of cloud content caching approaches that of cluster content caching, e.g., $\frac{\theta^{\mathrm{T}}_l}{\theta^{\mathrm{C}}_l} \rightarrow 1$, it requires that the backhaul link capacity goes infinity. Therefore, the deployment of cluster content caches can migrate the loadings of backhaul links, and improve the QoS guarantees.

{Although some non-cacheable content cannot be stored locally, our studied cluster caching structure can still improve QoS guarantees with low power consumption in practical wireless networks, where both cacheable and non-cacheable content is requested. In particular, by storing some popular cacheable content in local cluster caches, the delay caused by obtaining non-cacheable content from cloud content cache can be reduced due to the loading mitigation of the links between local cloud centers and the data centers, and thus better QoS guarantees of non-cacheable content can be supported.}

\section{Performance Analysis Based on Stochastic Geometry}
To evaluate the performance of cluster content caching in C-RANs, the effective capacity is studied in this section. To provide an accurate analytical model, the locations of RRHs are modeled as a homogenous Poisson point process (PPP) $\Psi_{\mathrm{R}}$ with a given density $\lambda_\mathrm{R}$. The locations of users are modeled as a homogenous marked PPP $\Phi_\mathrm{u}(M_n)$ with a given density $\lambda_\mathrm{u}$, where the mark $M_n$ is the type of content that the $n$-th user $\mathrm{U}_n$ in $\Phi_\mathrm{u}(M_n)$ requests, and required content objects of users are independent with each other. Compared with other complicated point processes, homogeneous PPP is more suitable for characterizing the seamless coverage of C-RANs, while the edge effect exists in the clustering point process. Moreover, the difference between homogeneous PPP and clustering point process is their densities of points only. Except the derivation of density functions, the analysis of all these point processes will follow a similar mathematical paradigm, and our derived analytical results can be extended to other PPP model by replacing the density functions. Note that the caching coordinations are implemented by sharing a common cache in each cluster, while the coordinations of signal processing, such as coordinated multiple point (CoMP) transmissions and network beamforming, are not considered. {The reason that we consider a simple case without joint signal processing is to provide some meaningful insights into content caching in C-RANs. }

\subsection{Effective Capacity of A Typical User Associated with A Specific RRH}
Each user in $\mathcal{C}_T$ can be treated as a typical user since its co-channel interference follows an identical distribution based on the introduced PPP model. Without loss of generality, we consider $\mathrm{U}_i$ in $\mathcal{C}_T$ as a typical user, which accesses a given RRH that serves its requested content.

{Recalling \eqref{eqn:effect2}, the key step of analyzing the effective capacity is to study the expectation of $Z = (1 + \gamma)^{- \mu \theta_j \bar{T}}$, which can be derived by using the following approximation. In particular, the range of the receive signal-to-interference-plus-noise ratio (SINR) $\gamma$ can be divided into $N$ disjoint intervals, e.g., the $n$-th intervals can be expressed as $I_n = [\gamma_{n}, \gamma_{n+1})$, $n = 1, \ldots, N$, $0 = \gamma_1 < \cdots < \gamma_{N+1} = \gamma_{\max} < \infty$, and $\gamma$ can be represented by the $n$-th typical value $\bar{\gamma}_n$ when $\gamma \in I_n$. Then $\mathbb{E}\{Z\}$ can be approximately expressed as
\begin{equation}\label{eqn:expec z}
\mathbb{E}\{Z\} \approx \sum_{n = 1}^N \big(\Pr \{\gamma < \gamma_{n + 1}\} - \Pr \{\gamma < \gamma_n\} \big)\big(1 + \bar{\gamma}_n\big)^{ - \mu \theta_j W \bar T}.
\end{equation}
This approximation is similar to the scalar quantization of analog signals, and the optimum setting of a typical value of $\gamma$ in $I_n$ is $\bar{\gamma}_n = (\gamma_{n} + \gamma_{n+1}) / 2$ based on quantization theory \cite{b41}. Based on \eqref{eqn:expec z}, a tractable expression for effective capacity can be given as in the following theorem.}
\begin{theorem}
Considering a typical user $\mathrm{U}_i$, which accesses a specific RRH $\mathrm{R}_m$ that provides its required content $S_j$, its effective capacity can be expressed as
\begin{eqnarray}\label{eqn:eff cap g}
E_{i,m}(\theta_j, d_m) = - \frac{1}{\theta_j W \bar{T}} \ln (\mathcal{G}(\theta_j, d_m)),
\end{eqnarray}
where $\theta_j$ is the QoS exponent of $S_j$, $d_m$ is the distance between $\mathrm{U}_i$ and $\mathrm{R}_m$, $\mathcal{G}(\theta_j, d_m)$ is given as \eqref{eqn:G} on the following page,
\begin{figure*}[ht]
\begin{eqnarray}\label{eqn:G}
\mathcal{G}(\theta_j, d_m) = \sum_{n = 1}^N \bigg( e^{- 2 \pi A(\beta) \gamma_n^{\frac{2}{\beta}} \lambda_\mathrm{R} d_m^2 - \frac{\gamma_n d_m^\beta \sigma ^2}{\rho}} - e^{- 2 \pi A(\beta) \gamma_{n + 1}^{\frac{2}{\beta}} \lambda_\mathrm{R} d_m^2 - \frac{\gamma_{n + 1} d_m^\beta \sigma ^2}{\rho}}\bigg) \big(1 + \bar{\gamma}_n \big)^{- \mu \theta_j \bar{T}}
\end{eqnarray}
\hrulefill
\end{figure*}
where $A(\beta) = \frac{1}{\beta} \Gamma(\frac{2}{\beta}) \Gamma(1 - \frac{2}{\beta})$, and $\Gamma(x)$ denotes the gamma function.
\end{theorem}
\begin{proof}
To obtain a tractable expression for typical user effective capacity, the outage probability $\Pr\{\gamma < \gamma_n\}$ is required. By following a similar mathematical paradigm in \cite{b42}, it can be derived as
\begin{eqnarray}\label{eqn:eff cap1}
\nonumber && \Pr \big\{\gamma < \gamma_n\big\} \\
\nonumber &=& \mathbb{E}_{\Psi_{\mathrm{R}}, \mathrm{R}_j \in \Psi_\mathrm{R} / \{\mathrm{R}_m\}} \bigg\{ \Pr \bigg\{ \big| h_m \big|^2 < \frac{\gamma_n d_m^\beta}{\rho} \big( \mathcal{I} + \sigma ^2 \big) \bigg\} \bigg\} \\
\nonumber &=& 1 - \mathbb{E}_{\Psi_{\mathrm{R}}, \mathrm{R}_j \in \Psi_\mathrm{R} / \{\mathrm{R}_m\}} \bigg\{ e^{ - \frac{\gamma_i d_m^\beta}{\rho} \big(\mathcal{I} + \sigma^2\big)}\bigg\} \\
\nonumber &=& 1 - \underbrace{e^{-\frac{{\gamma_n} d_m^{\beta} \sigma^2}{\rho}} \mathbb{E}_{\Psi_\mathrm{R}} \bigg\{\prod_{\mathrm{R}_j \in \Psi_\mathrm{R} / \{\mathrm{R}_m\}} \frac{1}{1 + {\gamma_n} d_m^{\beta} d_j^{-\beta}} \bigg\}}_{\mathcal{K}_1},\\
\end{eqnarray}
where $\mathcal{I} = \sum_{\mathrm{R}_j \in \Psi_{\mathrm{R}} / \{\mathrm{R}_m\}} \rho d_j^{-\beta} |h_j|^2$ denotes the co-channel interference, and $|h_i|^2$ of each channel follows independent and identically exponential distribution. Based on the probability generating functional (PGFL) of PPP, ${\mathcal{K}_1}$ can be expressed as
\begin{eqnarray}\label{eqn:eff cap5}
{\mathcal{K}_1} = \prod_{l=1}^L e^{-\frac{{\gamma _i} d_m^{\beta} \sigma^2}{\rho}} \mathcal{J}_1,
\end{eqnarray}
where
\begin{eqnarray}\label{eqn:eff cap3}
\nonumber \mathcal{J}_1 &=& \exp \bigg[-2 \pi \lambda_{\mathrm{R}} \int_0^\infty \bigg(1 - \frac{1}{1 + {\gamma_n} d_m^{\beta} d_j^{-\beta}} \bigg) d_j \mathrm{d} d_j\bigg]\\
\nonumber &=& \exp \bigg(-2 \pi \lambda_\mathrm{R} {\gamma_n}^{\frac{2}{\beta}} d_m^2 \int_0^\infty \frac{y}{y^\beta + 1} \mathrm{d} y\bigg) \\
&=& e^{-2 \pi \lambda_\mathrm{R} A(\beta) {\gamma_n}^{\frac{2}{\beta}} d_m^2}.
\end{eqnarray}
Based on \eqref{eqn:eff cap1}, \eqref{eqn:eff cap5},and \eqref{eqn:eff cap3}, $\Pr \{ \gamma < \gamma_n\}$ can be given as
\begin{eqnarray}\label{eqn:eff cap6}
\Pr \{\gamma < \gamma_n\} = 1 - e^{- 2 \pi A(\beta) \gamma_n^{\frac{2}{\beta}} \lambda_R d_m^2 - \frac{\gamma_n d_m^\beta \sigma ^2}{\rho}},
\end{eqnarray}
Substituting \eqref{eqn:eff cap6} into \eqref{eqn:expec z}, Theorem 1 is proved.
\end{proof}

{The approximation of \eqref{eqn:expec z} is an appropriate approach based on quantization theory, and has been widely applied in signal compression. The accuracy of quantization is mainly determined by the squared-error distortion, which can be improved by dividing the range of $\gamma$ into smaller intervals. In this paper, the quantization intervals are divided equally, and the simulation results in Section VI show that the analytical results match the Monte Carlo results perfectly when the maximum value of $\gamma$ is set as $\gamma_{\max} = 5 \times 10^4$ and the number of intervals is $N = 10^6$, respectively.}
\subsection{Average Effective Capacity of A Typical Cluster $\mathcal{C}_T$}
The average effective capacity is proposed as a metric to evaluate the performance of cluster content caching in C-RANs. In this part, each user is required to access its nearest RRH that provides its desired content object so that the received signal power can be maximized. {First, the hit ratio of the cluster content cache in $\mathcal{C}_T$ is denoted as $P_{\mathrm{hit}}$, which characterizes the probability that the requests can be satisfied at $\mathcal{U}_T$ \cite{b44}
\begin{eqnarray}\label{eqn:eff cap7}
P_{\mathrm{hit}} = \frac{\mathrm{Number~of~requests~served~by}~\mathcal{U}_T}{\mathrm{Total~number~of~user~requests}} = \sum_{l \in \mathcal{U}_T} P_l,
\end{eqnarray}
where $P_l$ is the probability that the content object $S_l$ is requested by the users, which can capture the popularity of $S_l$. Note that the summation of popularity ratios is 1, and thus the hit ratio follows the constraint $P_{\mathrm{hit}} = \sum_{l \in \mathcal{U}_T} P_l \leqslant \sum_{l \in \Omega_C} P_l =1$, which is a probability measure.} And the average effective capacity of content object $S_l$ can be expressed as
\begin{eqnarray}\label{eqn:eff cap8}
\bar{E}_l = P_{\mathrm{hit}} \bar{E}(\theta^{\mathrm{T}}_l) + (1 - P_{\mathrm{hit}}) \bar{E}(\theta^{\mathrm{C}}_l),
\end{eqnarray}
where $\bar{E}(\theta^{\mathrm{T}}_l)$ and $\bar{E}(\theta^{\mathrm{C}}_l)$ are the average effective capacity of content object $S_l$ that is served by $\mathcal{U}_T$ and $\mathcal{U}_C$, respectively. Due to the differences in the popularity of content objects, the average effective capacity of $\mathcal{C}_T$ can be given as $\bar{E}_T = \sum_{l = 1}^L P_l \bar{E}_l$. Moreover, RRHs belonging in $\Psi_\mathrm{R}$ can be divided into $L$ disjoint partitions to serve $L$ different content objects. The formulation of a RRH set $\Psi_l$ that serves $S_l$ can be treated as an independent thinning process of a homogenous PPP $\Psi_\mathrm{R}$, which is also a homogenous PPP with a fixed density $\lambda_l$, and $\sum_{l = 1}^L \lambda_l = \lambda_{\mathrm{R}}$.

Based on the aforementioned assumptions and the results given in Theorem 1, tractable expressions for $\bar{E}(\theta^{\mathrm{T}}_l)$ and $\bar{E}(\theta^{\mathrm{C}}_l)$ in \eqref{eqn:eff cap8} can be obtained, and the following corollary of average cluster effective capacity is provided.
\begin{corollary}
When each user accesses to the nearest RRH that serves its desired content object, the average effective capacity of a typical cluster $\mathcal{C}_T$ can be written as
\begin{eqnarray}\label{eqn:eff cap9}
\bar{E}_T = P_{\mathrm{hit}} \sum_{l=1}^L \bar{E}(\theta^{\mathrm{T}}_l) + (1 - P_{\mathrm{hit}}) \sum_{l=1}^L \bar{E}(\theta^{\mathrm{C}}_l),
\end{eqnarray}
where $\bar{E}(\theta^{\mathrm{T}}_l)$ and $\bar{E}(\theta^{\mathrm{C}}_l)$ can be expressed as
\begin{eqnarray}\label{eqn:eff cap10}
\nonumber \bar{E}(\theta_l) = P_l \sum_{n = 1}^N \big[\mathcal{L}_l(\gamma_n) - \mathcal{L}_l(\gamma_{n+1})\big] \big(1 + \bar{\gamma}_n\big)^{ - \mu \theta_l W \bar{T}},\\
\theta_l = \theta_l^\mathrm{T},\theta_l^\mathrm{C},
\end{eqnarray}
$P_l$ denotes the popularity of $S_l$, and $\mathcal{L}_l(\gamma_n)$ can be given as \eqref{eqn:eff cap11} on the following page,
\begin{figure*}[ht]
\begin{eqnarray}\label{eqn:eff cap11}
\mathcal{L}_l(\gamma_n) = 1 - 2 \pi \lambda_l \int_0^\infty d_m e^{ - ( 2 \pi A(\beta) \gamma_n^{\frac{2}{\beta}} (\lambda_\mathrm{R} - \lambda_l) + \pi \lambda_l u(\gamma_n, \beta) + \pi \lambda_l) d_m^2} e^{ - \frac{\gamma_n d_m^\beta \sigma ^2}{\rho}} \mathrm{d} d_m
\end{eqnarray}
\hrulefill
\end{figure*}
where $u(\gamma_n, \beta) = \gamma_n^{{2}/ {\beta}} \int_{\gamma_n^{- {2}/ {\beta}}}^\infty (1 + x^{\beta / 2})^{-1} \mathrm{d} x$. In an interference limited C-RAN scenario, a tractable expression for $\mathcal{L}_l(\gamma_n)$ is given by
\begin{eqnarray}\label{eqn:eff cap12}
\mathcal{L}_l(\gamma_n) = 1 - \frac{1}{2 A(\beta) \gamma_n^{2 / \beta} (q_l - 1) + u(\gamma_n, \beta) + 1},
\end{eqnarray}
where $q_l = \lambda_\mathrm{R} / \lambda_l$.
\end{corollary}
\begin{proof}
Due to the definition of $\bar{E}(\theta^{\mathrm{T}}_l)$ and $\bar{E}(\theta^{\mathrm{C}}_l)$ in \eqref{eqn:eff cap8}, they can be expressed as $\bar{E}(\theta_l) = P_l \mathbb{E}_{d_m}\{E_{i,m}(\theta_l, d_m)\}$, $\theta_l = \theta_l^\mathrm{T},\theta_l^\mathrm{C}$. Then the key step is to derive a tractable expression for $\mathcal{L}_l(\gamma_n)$ in \eqref{eqn:eff cap10}. Unlike Theorem 3, each user is required to access its nearest serving RRH, which establishes a constraint of the locations of RRHs serving $S_l$. Then $\mathcal{K}_1$ in \eqref{eqn:eff cap5} should be expressed as
\begin{eqnarray}\label{eqn:ave eff cap1}
{\mathcal{K}_1} = e^{-\frac{{\gamma _i} d_m^{\beta} \sigma^2}{\rho}} \mathcal{J}_2 \prod_{k \neq l} \mathcal{J}_1,
\end{eqnarray}
where $\mathcal{J}_1$ is given in \eqref{eqn:eff cap3}, and $\mathcal{J}_2$ can be expressed as
\begin{eqnarray}\label{eqn:ave eff cap2}
\nonumber \mathcal{J}_2 &=& \exp \bigg[-2 \pi \lambda_{\mathrm{R}} \int_{d_j}^\infty \bigg(1 - \frac{1}{1 + {\gamma_n} d_m^{\beta} d_j^{-\beta}} \bigg) d_j \mathrm{d} d_j\bigg]\\
&=& e^{ - \pi \lambda_l u(\gamma_n, \beta) d_m^2}.
\end{eqnarray}
The probability density function (PDF) of $d_m$ can be given as $f(d_m) = 2 \pi \lambda _i d_m e^{ - \pi \lambda_i d_m^2}$, and  $\mathcal{L}_l(\gamma_n)$ given as \eqref{eqn:eff cap11} can be obtained. In an interference limited scenario, the impacts of noise can be ignored, and thus $\sigma^2$ in \eqref{eqn:eff cap11} can be set as zero. Then $\mathcal{L}_l(\gamma_n)$ can be written as \eqref{eqn:eff cap12}. And Corollary 4 has been proved.
\end{proof}

When $P_{\mathrm{hit}} = 0$, $\bar{E}_T = \sum_{l=1}^L \bar{E}(\theta^{\mathrm{C}}_l)$ is the effective capacity of no-cluster content caching strategy in C-RANs. As introduced previously, the QoS exponents follow a constraint $\theta^{\mathrm{L}}_l \leqslant \theta^{\mathrm{C}}_l$. Then the performance improvement of effective capacity achieved by cluster content caching structure in C-RANs can be written as \eqref{eqn:ave eff cap} on the following page based on Corollary 4.
\begin{figure*}[ht]
\begin{eqnarray}\label{eqn:ave eff cap}
\nonumber \Delta \bar{E}_T &=& \bar{E}_T - \bar{E}_T|_{P_{\mathrm{hit}} = 0} = P_{\mathrm{hit}} \bigg[\sum_{l=1}^L \big(\bar{E}(\theta^{\mathrm{T}}_l) - \bar{E}(\theta^{\mathrm{C}}_l)\big)\bigg] \\
&=& P_{\mathrm{hit}} \bigg\{\sum_{l=1}^L P_l \bigg[\sum_{n = 1}^N \big(\mathcal{L}_l(\gamma_n) - \mathcal{L}_l(\gamma_{n+1})\big) \left(\big(1 + \bar{\gamma}_n\big)^{ - \mu \theta^{\mathrm{T}}_l \bar{T}} - \big(1 + \bar{\gamma}_n\big)^{ - \mu \theta^{\mathrm{C}}_l \bar{T}}\right)\bigg]\bigg\}.
\end{eqnarray}
\hrulefill
\end{figure*}
As shown in \eqref{eqn:ave eff cap}, the performance gains $\Delta \bar{E}_T$ increases as the hit ratio of local cluster caching $P_{\mathrm{hit}}$ increases. In particular, when all content can be served by the cluster content cache $\mathcal{U}_\mathrm{T}$, e.g., $M = L$ in \eqref{eqn:eff cap7}, the local content caching can approach its performance limits. Assuming that all the required content objects have the same QoS exponents of cloud and cluster content caching structures in $\mathcal{C}_T$, $\mathcal{U}_T$ should choose to keep the most popular content objects in its coverage area, which contribute more to the improvement of effective capacity.

{\subsection{Energy Efficiency Performance of A Typical Cluster $\mathcal{C}_T$}
The average effective capacity to the average power consumption ratio in a typical cluster $\mathcal{C}_T$ with a cluster content cache is defined as follows:
\begin{equation}\label{eqn:en_eff_clucac}
\eta_T = \frac{\bar{E}_T}{\bar{P}_{T}} = \frac{P_{\mathrm{hit}} \sum_{l=1}^L \bar{E}(\theta^{\mathrm{T}}_l) + (1 - P_{\mathrm{hit}}) \sum_{l=1}^L \bar{E}(\theta^{\mathrm{C}}_l)}{\lambda_{\mathrm{R}} \pi r_T^2 P_{\mathrm{R}} + K P_{\mathrm{CC}} + (1 - P_{\mathrm{hit}}) P_{\mathrm{BH}}},
\end{equation}
where $\bar{E}_T$ denotes the average effective capacity of $\mathcal{C}_T$ given in Corollary 4, $\bar{P}_{T}$ is the average total power consumption with cluster content caching, $P_{\mathrm{R}}$ denotes the power consumption of radio transmission and baseband processing for each RRH, $r_T$ is the radius of considered cluster $\mathcal{C}_T$, $P_{\mathrm{CC}}$ is the power consumption of keeping a content object in $\mathcal{U}_T$, $K$ is the size of $\mathcal{U}_T$, and $P_{\mathrm{BH}}$ is the power consumption of obtaining desired content objects through backhaul. Recalling \eqref{eqn:en_eff_clucac}, $\eta_T$ can characterize the energy efficiency with a specific QoS requirement. Similarly, the average effective capacity to the average power consumption ratio without cluster content caching can be expressed as follows for $P_{\mathrm{hit}} = 0$:
\begin{eqnarray}
\eta_T|_{P_{\mathrm{hit}} = 0} = \frac{\bar{E}_T |_{P_{\mathrm{hit}} = 0}}{\bar{P}_{T}|_{P_{\mathrm{hit}} = 0}} = \frac{\sum_{l=1}^L \bar{E}(\theta^{\mathrm{C}}_l)}{\lambda_{\mathrm{R}} \pi r_T^2 P_{\mathrm{R}} + P_{\mathrm{BH}}}.
\end{eqnarray}
And the comparison of power consumption between two different structures is quantified by
\begin{eqnarray}\label{eqn:comp_pow}
\bar{P}_{T} - \bar{P}_{T}|_{P_{\mathrm{hit}} = 0} = K P_{\mathrm{CC}} - P_{\mathrm{hit}} P_{\mathrm{BH}} \leqslant 0.
\end{eqnarray}
The last inequity in \eqref{eqn:comp_pow} follows the fact $P_{\mathrm{CC}} \ll P_{\mathrm{BH}}$, e.g., $P_{\mathrm{BH}}$ is larger than 10 W \cite{b70}, while $P_{\mathrm{CC}}$ is smaller than 1 W \cite{b19}. Moreover, recalling \eqref{eqn:ave eff cap}, it shows that our studied cluster content caching structure can achieve higher effective capacity. Therefore, employing cluster caches is an energy efficient approach to improving QoS guarantees in C-RANs.}

\section{Joint Design of RRU Allocation and RRH Association Based on a Nested Coalition Formation Game}

Based on the study in Section \uppercase\expandafter{\romannumeral3}, the QoS guarantees of cluster content caching are impacted by the radio resource allocation strategy, such as the serving RRH association and the RRU allocation. The management of the two categories of radio resource can be modeled as a coupled integer programming problem, which is an NP-hard problem. To provide an efficient solution, both serving RRH association and RRU allocation can be formulated as coalition formation games, and the joint design of them can be modeled as a nested coalition formation game.

\subsection{Serving RRH Association Strategy}

Assuming that there exists a subset of content set $\Omega_C$, e.g., $\Omega_j = \{S_{j_1},\ldots,S_{j_M}\}$, whose members share the same RRU $\mathrm{RU}_j$, then RRHs in $\mathcal{C}_T$ will form $M$ disjoint sets, which are denoted as $\mathcal{R}_{j_1},\ldots,\mathcal{R}_{j_M}$ to transmit the corresponding content objects in $\Omega_j$. Denoting $\mathcal{W}_{j_m}$ as the set of users that require $S_{j_m}$ in $\mathcal{C}_T$, then the expected effective capacity of content object $S_{j_m}$, which is served by RRHs in $\mathcal{R}_{j_m}$, can be expressed as follows when the location information of the members in $\mathcal{R}_{j_m}$ and $\mathcal{W}_{j_m}$ are available:
\begin{eqnarray}\label{eqn:eff cap cont}
\bar{E}(\mathcal{R}_{j_m}) =
\begin{cases}
& \sum_{\mathrm{U}_i \in \mathcal{W}_{j_m}} E(\theta_{j_m}^\mathrm{T}, d_i),~S_{j_m}~\mathrm{is~in}~\mathcal{U}_T \\
& \sum_{\mathrm{U}_i \in \mathcal{W}_{j_m}} E(\theta_{j_m}^\mathrm{C}, d_i),~S_{j_m}~\mathrm{is~in}~\mathcal{U}_C
\end{cases},
\end{eqnarray}
where $E(\theta_{j_m}^\mathrm{T}, d_i)$ and $E(\theta_{j_m}^\mathrm{C}, d_i)$ follow the effective capacity of a specific user given in Theorem 3, $d_i$ denotes the distance between $\mathrm{U}_i$ and its serving RRH in $\mathcal{R}_{j_m}$, and $\theta_{j_m}^\mathrm{T}$ and $\theta_{j_m}^\mathrm{C}$ are the QoS exponents of content object $S_{j_m}$ served by local cluster cache and cloud cache, respectively. As shown in \eqref{eqn:eff cap cont}, the QoS guarantees of users in $\mathcal{W}_{j_m}$ are mainly determined by RRHs in $\mathcal{R}_{j_m}$, which have great impacts on the path loss that characterized by $d_i$.
\subsubsection{Utility Function Formulation}
The RRHs can be treated as a scarce resource in the studied C-RAN scenario, which are competed for by the users to improve their QoS guarantees. Therefore, the serving RRH association can be modeled as a coalition formation game, in which RRHs in $\mathcal{C}_T$ can be treated as players, and each RRH association result can be denoted as a partition $\Pi = \{\mathcal{R}_{j_1}, \ldots, \mathcal{R}_{j_M}\}$. The payoff of the $k$-th RRH $\mathrm{R}_k$ to join $\mathcal{R}_{j_m}$ can be defined as the effective capacity improvement achieved by its participation, and the corresponding utility function is given as \eqref{eqn:uti fun rrh} on the following page,
\begin{figure*}[ht]
\begin{eqnarray}\label{eqn:uti fun rrh}
\phi_k(\mathcal{R}_{j_m}) =
\begin{cases}
& \bar{E}(\mathcal{R}_{j_m} \cup \mathrm{R}_k) - \bar{E}(\mathcal{R}_{j_m}) - \underbrace{c_{\mathrm{RH}} \bigg(P_{\mathrm{R}} + \frac{1}{\mathcal{O}(\mathcal{R}_{j_m})} P_{\mathrm{CC}}\bigg)}_{\mathrm{The~cost~part}~\tau_{k}},~S_{j_m}~\mathrm{is~in}~\mathcal{U}_T \\
& \bar{E}(\mathcal{R}_{j_m} \cup \mathrm{R}_k) - \bar{E}(\mathcal{R}_{j_m}) - \underbrace{c_{\mathrm{RH}} \bigg(P_{\mathrm{R}} + \frac{1}{\mathcal{O}(\mathcal{R}_{j_m})} P_{\mathrm{BH}}\bigg)}_{\mathrm{The~cost~part}~\tau_{k}},~S_{j_m}~\mathrm{is~in}~\mathcal{U}_C
\end{cases}
\end{eqnarray}
\hrulefill
\end{figure*}
where $\mathcal{O}(A)$ denotes the number of elements in set $A$, $\tau_{k}$ in \eqref{eqn:uti fun rrh} is the cost of $\mathrm{RRH}_k$ to join $\mathcal{R}_{j_m}$, which is determined by the total power consumption. In particular, $P_{\mathrm{R}}$ denotes the power consumption of a RRH spending on radio frequency and baseband processing, $P_{\mathrm{CC}}$ and $P_{\mathrm{BH}}$ follow the denotation of \eqref{eqn:en_eff_clucac}, and $c_{\mathrm{RH}}$ is the energy efficiency coefficient to control the impact of the cost part $\tau_{k}$. Note that each content object is required by multiple RRHs in $\mathcal{C}_T$, and thus the cost of acquiring it should be split among the corresponding RRHs. Then the total utility of each coalition can be written as
\begin{eqnarray}\label{eqn:eff cap clu}
v(\mathcal{R}_{j_m}) = \sum_{\mathrm{R}_k \in \mathcal{R}_{j_m}} \phi_k(\mathcal{R}_{j_m}).
\end{eqnarray}
\subsubsection{A Distributed Coalition Formation Algorithm}
Based on the utility function given in \eqref{eqn:uti fun rrh}, the payoff of $\mathrm{R}_k$ to join $\mathcal{R}_{j_m}$ is determined by only the members in $\mathcal{R}_{j_m}$. Thus the proposed RRH association problem can be modeled as a hedonic coalitional formation game, and the coalition formation process is accomplished by using preference relations. Based on the definitions in \cite{b26}, the preference relation indicates that a RRH strictly prefers to join a specific coalition rather than the other one. For example, consider two coalitions $\mathcal{R}_{j_m}$ and $\mathcal{R}_{j_n}$, $\mathrm{R}_k$ will choose to join $\mathcal{R}_{j_m}$ if the two given coalitions satisfy the preference relation $\mathcal{R}_{j_m} \succ_k \mathcal{R}_{j_n}$. Due to \cite{b26}, the criteria of preference relation determination can be established as follows:
\begin{eqnarray}\label{eqn:pref rela}
\nonumber && \mathcal{R}_{j_m} \succ_k \mathcal{R}_{j_n} \Leftrightarrow \\
\nonumber \mathrm{C1}: && \phi_k(\mathcal{R}_{j_m}) > \phi_k(\mathcal{R}_{j_n}) \\
\nonumber \mathrm{C2}: && v(\mathcal{R}_{j_m}) + v(\mathcal{R}_{j_n} \backslash \{\mathrm{R}_k\}) > \\
&& v(\mathcal{R}_{j_n}) + v(\mathcal{R}_{j_m} \backslash \{\mathrm{R}_k\}),
\end{eqnarray}
where C1 implies that $\mathrm{R}_k$ pursues higher individual payoff, while C2 can avoid the loss of total utility caused by the move of $\mathrm{R}_k$.
\begin{algorithm}[t!]
\caption{(A hedonic coalition formation-based serving RRH association algorithm)} \label{alg:rrh}
\begin{algorithmic}
\STATE \textbf{Initialization}:
For a given subset of content objects $\Omega_j^{\mathrm{sub}} = \{S_{j_1},\ldots,S_{j_M}\}$, all RRHs in $\mathcal{C}_T$ can be divided into $M$ disjoint coalitions $\mathcal{R}_{j_1}, \ldots , \mathcal{R}_{j_M}$;
\STATE \textbf{Repeat}:
For each RRH coalition $\mathcal{R}_{j_m}$
 \begin{enumerate}
 \item Each RRH in $\mathcal{R}_{j_m}$, e.g., $\mathrm{R}_k$, negotiates with other coalitions $\mathcal{R}_{j_n}$, $j_m \neq j_n$, and obtains the individual and the total coalition utility values based on \eqref{eqn:uti fun rrh} and \eqref{eqn:eff cap clu}, respectively;
 \item If the obtained utility values satisfy the preference relation criteria given in \eqref{eqn:pref rela}, then $\mathrm{R}_k$ leaves $\mathcal{R}_{j_m}$ and joins $\mathcal{R}_{j_n}$;
 \end{enumerate}
\STATE\textbf{Termination}: When the members of $\mathcal{R}_{j_1}, \ldots , \mathcal{R}_{j_M}$ do not change.
\end{algorithmic}
\end{algorithm}

The conditions in \eqref{eqn:pref rela} ensure that each C-RAN cluster can achieve considerable QoS performance gains via a coalition formation game-based RRH association strategy, and the corresponding distributed RRH association algorithm can be described in Algorithm \ref{alg:rrh}. In particular, RRHs are assigned to serve the content objects $S_{j_1},\ldots,S_{j_M}$ in the initialization phase, and RRHs take turns to negotiate with other potential coalitions, and determine whether to join a new one based on the criteria given in \eqref{eqn:pref rela} until the obtained partition results converge. The following proposition asserts that the proposed algorithm will converge to a Nash stable partition.
\begin{proposition}(Theorem 1, \cite{b26})
Starting from any initial partition, Algorithm 1, which is based on a hedonic solution, will always converge to a Nash stable partition.
\end{proposition}
\subsection{RRU Allocation Strategy}
Besides the serving RRHs association, the performance of radio-access channels is also impacted by the RRU allocation strategy. Similar to RRH association, RRU allocation can also be formulated as a coalition formation game. In particular, the content objects $S_1,\ldots,S_L$ can be treated as players, which negotiate with each other to form disjoint coalitions, and the content objects in the same coalition share a common RRU.
\subsubsection{Utility Function Formulation}
Without loss of generality, we focus on a coalition $\mathcal{T}_i = \{S_{i_1}, \ldots, S_{i_F}\}$, whose members are the content objects sharing the $i$-th RRU. The payoff of coalition $\mathcal{T}_i$ can be defined as the sum effective capacity of the considered cluster $\mathcal{C}_T$, Then the corresponding utility function can be expressed as \eqref{eqn:eff cap rrb} on the following page,
\begin{figure*}[ht]
\begin{eqnarray}\label{eqn:eff cap rrb}
\psi(\mathcal{T}_i) =
\bigg[\sum_{S_{i_n} \in \mathcal{T}_i} \bar{E}(\mathcal{R}_{i_n}) - \underbrace{c_{\mathrm{RU}} \bigg(\sum_{S_{i_n} \in \mathcal{T}_i} \mathcal{O}(\mathcal{R}_{i_n}) P_{\mathrm{R}} + \mathcal{O}(\mathcal{U}_T) P_{\mathrm{CC}} + \mathcal{O}(\mathcal{U}_C) P_{\mathrm{BH}}\bigg)}_{\mathrm{The~cost~part}~\varrho_{i}}\bigg]^+
\end{eqnarray}
\hrulefill
\end{figure*}
where $\bar{E}(\mathcal{R}_{i_n})$ is the expected effective capacity of content object $S_{i_n}$ given in \eqref{eqn:eff cap cont}, $\varrho_{i}$ is the cost part of forming coalition $\psi(\mathcal{T}_i)$, $c_{\mathrm{RU}}$ is an energy efficiency coefficient to control the impact of cost in \eqref{eqn:eff cap rrb}, and $(a)^+ = \max(a,0)$.

\subsubsection{A Distributed Merge and Split Algorithm}
Unlike the serving RRH association problem, the payoff of each content in RRU allocation is not related only to the members in its own coalition, since the spectral efficiency is determined by the partition result via the coalition formation process. Therefore, the hedonic coalition formation algorithm in the previous subsection is not applicable. Moreover, due to the existence of costs, our studied game is non-superadditive with an empty core, and thus the grand coalition cannot be formed \cite{b27}. To solve the RRU allocation problem efficiently, a distributed merge and split method is studied in this subsection.

The key idea of the proposed algorithm is to form the coalitions of content objects only through merge and split operations. Without loss of generality, we take the first $l$ coalitions $\mathcal{T}_1,\ldots,\mathcal{T}_l$ as an example, and the rules for merge and split operations are shown as follows:
\begin{itemize}
\item \textit{Merge rule:} If the utility functions of coalitions $\mathcal{T}_1,\ldots,\mathcal{T}_l$ satisfy $\sum_{j=1}^l \psi \big(\mathcal{T}_j\big) < \psi \big(\cup_{j=1}^l \mathcal{T}_j\big)$, then $\mathcal{T}_1,\ldots,\mathcal{T}_l$ merge as one cluster $\cup_{j=1}^l \mathcal{T}_j$.
\item \textit{Split rule:} If there exists a partition of a coalition $\mathcal{T}_j = \big\{\mathcal{P}_1, \ldots, \mathcal{P}_m\big\}$, whose utility functions satisfy $\psi \big(\mathcal{T}_j\big) < \sum_{i=1}^m \psi \big(\mathcal{P}_i\big)$, then $\mathcal{T}_j$ splits into $m$ disjoint coalitions $\mathcal{P}_1, \ldots, \mathcal{P}_m$.
\end{itemize}
\subsection{A Nested Coalition Formation Game-Based Algorithm}
Recalling \eqref{eqn:eff cap rrb}, the average sum effective capacity of content transmission in $\mathcal{C}_T$ is jointly determined by the RRH and the RRU allocations, and thus the two formulated coalitional games are coupled. To obtain a tractable method solving this problem, a nested coalition formation game-based algorithm is studied. In particular, the effective capacity of each content object $\bar{E}(\mathcal{R}_{j_m})$ given in \eqref{eqn:eff cap cont}, which is the payoff part in \eqref{eqn:eff cap rrb}, can be interpreted as the accumulation of effective capacity improvement achieve by RRHs joining the coalition $\bar{E}(\mathcal{R}_{j_m})$ gradually, e.g.,
\begin{eqnarray}\label{eqn:sum uti}
\nonumber \bar{E}(\mathcal{R}_{j_m}) &=& \bar{E}(\mathcal{R}_{j_m}) - \bar{E}(\varnothing) \\
&=& \sum_{\mathrm{R}_k \in \mathcal{R}_{j_m}} \big[\bar{E}(\mathcal{R}_{j_m} \cup \mathrm{R}_k) - \bar{E}(\mathcal{R}_{j_m})\big].
\end{eqnarray}
Such an interpretation coincides with the process of RRH association. Therefore, when the coefficients of cost control in \eqref{eqn:uti fun rrh} and \eqref{eqn:eff cap rrb} satisfy $c_{\mathrm{RH}} = c_{\mathrm{RU}} = c_0$, the following relationship between the utility functions of two studied coalitional games can be established:
\begin{equation}\label{eqn:rela uti fun}
\psi(\mathcal{T}_i) = \bigg[\sum_{S_{i_n} \in \mathcal{T}_i} v(\mathcal{R}_{i_n}) \bigg]^+ = \bigg[\sum_{S_{i_n} \in \mathcal{T}_i} \sum_{\mathrm{R}_k \in \mathcal{R}_{i_n}} \phi_k(\mathcal{R}_{i_n}) \bigg]^+.
\end{equation}
Equation \eqref{eqn:rela uti fun} shows that the utility of RRU allocation can be obtained by using the utility of RRH association. Moreover, due to the preference relation criteria given in \eqref{eqn:pref rela}, it ensures that the partition results of serving RRH association provide reliable utility for the coalitional game of RRU allocation. Then a nested coalitional formation can be proposed for the joint design of RRU allocation and RRH association. As shown in Algorithm \ref{alg:nest}, all the required content objects are randomly assigned to orthogonal RRUs, and form $N$ disjoint coalitions in the initialization phase. During the iteration phase, the content coalitions take arbitrary merge and split operations. The proposed hedonic coalitional game-based RRH association algorithm given in the previous part is nested to provide the total utility values of the studied C-RAN cluster $\mathcal{C}_T$ in each iteration. After obtaining a final partition result, each RRH in $\mathcal{C}_T$ should be check if it is the nearest serving RRH for any users, and those unnecessary ones are turned into sleep mode to improve energy efficiency.

Unlike the conventional coalition formation games, the utility value of the studied nested coalition formation game is jointly determined by the partition results of RRU allocation and RRH association, which may impact the convergent performance. For example, $\mathcal{T}_i$ and $\mathcal{T}_j$ merge as one coalition $\mathcal{T}_{i \cup j}$, since their utility functions follow the merge rule during a specific iteration. However, the Nash stable partition resulting from coalitional game-based RRH association algorithm may not be unique, and different initializations of RRH partitions or negotiation orders will lead to different utility values of content coalitions, which may change the relationship between $\psi(\mathcal{T}_i \cup \mathcal{T}_j)$ and $\psi(\mathcal{T}_i) + \psi(\mathcal{T}_j)$. Therefore, there exists the possibility that $\mathcal{T}_{i \cup j}$ may split into $\mathcal{T}_i$ and $\mathcal{T}_j$ once more, and the proposed algorithm may not converge. To avoid the instability of utility values, the initialization and the negotiation order should be given in the RRH association algorithm. In particular, during both the initialization and the iteration phases, the negotiation order of the considered set $\mathcal{B} = \{b_{k_1},\ldots,b_{k_L}\}$ is given by following the increasing order of the superscripts of its members, e.g.,
\begin{eqnarray}\label{eqn:nego ord}
\nonumber \Lambda : b_{k_1} \rightarrow \cdots \rightarrow b_{k_l} \rightarrow \cdots \rightarrow b_{k_L},~b_{k_1},\ldots,b_{k_L} \in \mathcal{B},\\
\mathrm{and}~k_1 < \cdots < k_l < \cdots < k_L.
\end{eqnarray}
Then the initializations of RRH partitions in Algorithm \ref{alg:nest} can be written as follows, in which RRHs are encouraged to join a coalition that can maximize its utility during the initialization phase of RRH association:
\begin{eqnarray}\label{eqn:init rrh}
\nonumber \mathrm{R}_k \in \mathcal{R}_{j_m},~s.t.~\phi_k(\mathcal{R}_{j_m}) = \max \limits_{S_{j_n} \in \Omega_C^j} (\phi_k(\mathcal{R}_{j_n})),\\
\mathrm{R}_k \in \mathcal{R} \times \Lambda,
\end{eqnarray}
where $\mathcal{R} \times \Lambda$ denotes an ordered set that generate based on $\mathcal{R}$ with a given order $\Lambda$ defined in \eqref{eqn:nego ord}. To verify the convergence performance of the proposed joint allocation algorithm, which is based on a nested coalition formation game, the following theorem is provided.
\begin{algorithm}[t!]
\caption{(A nested coalition formation game-based algorithm)} \label{alg:nest}
\begin{algorithmic}
 \STATE \textit{Step 1. Joint design of RRU allocation and RRH association}
 \STATE \textbf{Initialization}:
 Formulate $N$ disjoint coalitions of content objects $\mathcal{T}_1, \ldots , \mathcal{T}_N$, $1 \leqslant N \leqslant L$;
 \STATE \textbf{Repeat}:
 For each content coalition $\mathcal{T}_i ~ (\mathcal{T}_i \neq \varnothing)$
 \begin{enumerate}
 \item Merge operation: Negotiate with other coalitions, e.g., $\mathcal{T}_{j_1},\ldots,\mathcal{T}_{j_l}, ~j_1,\ldots,j_l \neq i$;
 \begin{itemize}
 \item Obtain the utility values of $\psi(\mathcal{T}_i)$, $\psi(\mathcal{T}_{j_1}),\ldots,\psi(\mathcal{T}_{j_l})$ and $\psi(\mathcal{T}_i \cup \mathcal{T}_{j_1} \cdots \cup \mathcal{T}_{j_l})$ based on \eqref{eqn:rela uti fun} and Algorithm \ref{alg:rrh}, in which the initial RRH partition is given in \eqref{eqn:init rrh}, and both the RRH coalitions and their members follow the negotiation order defined in \eqref{eqn:nego ord};
 \item If $\psi(\mathcal{T}_i) + \sum_{p=1}^l \psi(\mathcal{T}_{j_p}) < \psi(\mathcal{T}_i \cup \mathcal{T}_{j_1} \cdots \cup \mathcal{T}_{j_l})$, $\mathcal{T}_i = \big\{\mathcal{T}_i \cup \mathcal{T}_{j_1} \cdots \cup \mathcal{T}_{j_l}\big\}$, $\mathcal{T}_{j_1} = \cdots = \mathcal{T}_{j_l} = \varnothing$;
 \end{itemize}
 \item Split operation: For each subset $\mathcal{T}^{\mathrm{sub}}_{i_n}$ in $\mathcal{T}_i$
 \begin{itemize}
 \item Obtain the utility values of $\psi(\mathcal{T}_i)$, $\psi(\mathcal{T}^{\mathrm{sub}}_{i_n})$ and $\psi(\mathcal{T}_i / \mathcal{T}^{\mathrm{sub}}_{i_n})$ based on \eqref{eqn:rela uti fun} and Algorithm \ref{alg:rrh}, in which the initial RRH partition is given in \eqref{eqn:init rrh}, and both the RRH coalitions and their members follow the negotiation order defined in \eqref{eqn:nego ord};
 \item If $\psi(\mathcal{T}^{\mathrm{sub}}_{i_n}) + \psi(\mathcal{T}_i / \mathcal{T}^{\mathrm{sub}}_{i_n}) > \psi(\mathcal{T}_i)$, $\mathcal{T}_i = \mathcal{T}_i / \mathcal{T}^{\mathrm{sub}}_{i_n}$, and formulate a new coalition $\mathcal{T}_k = \mathcal{T}^{\mathrm{sub}}_{i_n}$;
 \end{itemize}
 \end{enumerate}
 \STATE

 \textbf{Termination}: When the members of each coalition do not change.
 \STATE \textit{Step 2. Identify RRHs that are not required by any user, and put them into sleep mode.}
\end{algorithmic}
\end{algorithm}
\begin{theorem}
The proposed joint allocation algorithm, which is obtained by solving a nested coalition formation game, converges to a $\mathbb{D}_{hp}$-stable partition, which means that all the possible partitions cannot recur with additional merge or split operations.
\end{theorem}
\begin{proof}
Assuming that the final partition result is not $\mathbb{D}_{hp}$-stable, there exists one partition $\Pi^* = \{\mathcal{T}_1^*, \ldots, \mathcal{T}_N^*\}$ that can recur during the iteration process, and the corresponding transformation can be expressed as
\begin{eqnarray}\label{eqn:trans ord}
\Pi^* \rightarrow \Pi_1 \rightarrow \cdots \rightarrow \Pi_i \cdots \rightarrow \Pi^*.
\end{eqnarray}
The total utility value of content transmissions keeps increasing via the merge and split operations due to the rules introduced previously. Therefore, based on the transitivity of coalition utility caused by the partition transform given in \eqref{eqn:trans ord}, there exists a subset $\{\mathcal{T}_{j_1}^*,\ldots,\mathcal{T}_{j_l}^*\}$ of $\Pi^*$, and the utility values of its members satisfy the following relationship:
\begin{eqnarray}\label{eqn:rela rrb}
\sum_{p = 1}^l \psi_1(\mathcal{T}_{j_p}^*) < \sum_{p = 1}^l \psi_2(\mathcal{T}_{j_p}^*),
\end{eqnarray}
where $\psi_1(\mathcal{T}_{j_p}^*)$ denotes the utility of $\mathcal{T}_{j_p}^*$ when $\Pi^*$ first shows up in \eqref{eqn:trans ord}, and $\psi_2(\mathcal{T}_{j_p}^*)$ is defined similarly for the second appearance of $\Pi^*$. Due to \eqref{eqn:rela rrb}, there exists at least one content coalition $\mathcal{T}_{j_q}^*$, whose utility follows the relationship $\psi_1(\mathcal{T}_{j_q}^*) < \psi_2(\mathcal{T}_{j_q}^*)$, $\mathcal{T}_{j_q}^* \in \{\mathcal{T}_{j_1}^*,\ldots,\mathcal{T}_{j_l}^*\}$. The spectral efficiency stays the same during the two appearances of $\Pi^*$, since the RRU allocation results are identical. Then the difference between $\psi_1(\mathcal{T}_{j_q}^*)$ and $\psi_2(\mathcal{T}_{j_q}^*)$ must be caused by different serving RRH associations, which can be denoted as $\Xi_1 = \{\mathcal{R}^1_{i_1}, \ldots, \mathcal{R}^1_{i_k}\}$ and $\Xi_2 = \{\mathcal{R}^2_{i_1}, \ldots, \mathcal{R}^2_{i_k}\}$, respectively. As described in Algorithm \ref{alg:nest}, both the initialization and the negotiation order have been given, and thus $\Xi_2$ can be obtained via limited steps of partition transforms from $\Xi_1$, e.g.,
\begin{eqnarray}\label{eqn:trans ord rrh}
\Xi_1 \rightarrow \Xi_{s_1} \rightarrow \cdots \rightarrow \Xi_2.
\end{eqnarray}
As shown in \eqref{eqn:trans ord rrh}, there must exist a RRH that wants to leave its current coalition and join a new one due to the preference relation defined in \eqref{eqn:pref rela}, and thus $\Xi_1$ is not Nash-stable. Such a result contradicts the conclusion given in Proposition 4. Therefore, Algorithm 2 converges to a $\mathbb{D}_{hp}$-stable partition result, and the proof is complete.
\end{proof}

Compared with centralized algorithms with global information, our proposed algorithm can allocate the serving RRHs and the RRUs distributively, and only local information is required. Therefore, our studied algorithm can achieve considerable performance gains with low computational complexity. {Moreover, to further improve the performance of cluster content caching, global joint processing, such as dynamic cluster formation and inter-cluster interference coordination scheduling, should be considered, as well as content correlation for the placement and management of both the cluster and the cloud caches. This global joint processing can be supported by our studied system model due to the existence of the global controller and the cloud cache at the cloud center.}
\section{A Suboptimal Algorithm with Lower Computational Complexity}
The nested structure of Algorithm \ref{alg:nest} requires that the total utility of content transmissions is obtained by solving a coalitional game-based RRH association problem in each iteration. To reduce the process complexity, an efficient method is to decouple the problem as two independent coalition formation games. In this section, a suboptimal RRU allocation algorithm is studied, whose utility does not depend on the partition results of the serving RRH association.
\subsection{Utility Function Formulation Based on Shapley Value}
To reduce the computational complexity, an alternative utility function for RRU allocation should be established. The interests of RRHs might be conflicting among the content objects, which should be considered for the RRU allocation. In particular, the content objects with less interest conflict on RRH association tend to share a common RRU. Therefore, to characterize the interest of a specific content object $S_i$ on each RRH, we need to evaluate the importance of each RRH on $S_i$, and the Shapley value is applied for the utility function formulation.

In our studied RRU allocation problem, the Shapley value can be interpreted as the expected marginal contribution of each RRH for the transmission of $S_i$. In particular, a grand coalition $\mathcal{N}$ formed by all RRHs in the typical C-RAN cluster $\mathcal{C}_T$ to transmit $S_i$, and the Shapley value of $\mathrm{R}_j$ is defined as its expected marginal contribution when it joins the grand coalition in a random order, which can be expressed as follows:
\begin{equation}\label{eqn:shap val}
v_{i,j} = \sum_{\mathcal{N}_s \subseteq \mathcal{N} \backslash \{\mathrm{R}_j\}} \frac{\mathcal{O}(\mathcal{N}_s)!(\mathcal{O}(\mathcal{N}) - \mathcal{O}(\mathcal{N}_s) - 1)!}{\mathcal{O}(\mathcal{N})} \mathcal{E}_{\mathcal{N}_s, \mathrm{RRH}_j},
\end{equation}
where $\mathcal{E}_{\mathcal{N}_s, \mathrm{RRH}_j} = \bar{E}(\mathcal{N}_s \cup \mathrm{RRH}_j) - \bar{E}(\mathcal{N}_s)$ is the marginal contribution of $\mathrm{R}_j$ in coalition $\mathcal{N}_s$ based on \eqref{eqn:eff cap cont}, and $\frac{\mathcal{O}(\mathcal{N}_s)!(\mathcal{O}(\mathcal{N}) - \mathcal{O}(\mathcal{N}_s) - 1)!}{\mathcal{O}(\mathcal{N})}$ is the probability that such a subset $\mathcal{N}_s$ occurs when RRHs join the grand coalition under a random order.

The importance of each RRH for a specific content object can be evaluated by its Shapley value, e.g., $\mathrm{R}_j$ with a higher Shapley value $v_{i,j}$ is more important to $S_i$. Then the utility function of RRU allocation can be formulated based on the interest conflicts among the content objects that share the same RRU, which can be expressed as follows for a given content object $S_{j}$ to join a coalition $\mathcal{T}_i = \{S_{i_1}, \ldots, S_{i_F}\}$, $S_{j} \notin \mathcal{T}_i$:
\begin{eqnarray}\label{eqn:shap val uti}
\varphi_{j}(\mathcal{T}_i) = \sum_{S_{i_n} \in \mathcal{T}_i} \bigg(\sum_{k=1}^K |v_{j,k} - v_{i_n,k}|\bigg) - \varrho_{i},
\end{eqnarray}
where the cost part $\varrho_{i}$ follows the definition given in \eqref{eqn:eff cap rrb}. In particular, when $v_{j,k} = v_{i_n,k}$ for each RRH, it implies that $S_{j}$ and $S_{i_n}$ have the same interests on RRH association, and the payoff is set as zero since it is the most competitive case of RRH association. As the differences of the Shapley values increase, the interest conflicts of RRHs between $S_{j}$ and $S_{i_n}$ can be mitigated, and thus the payoff of $S_{j}$ increases. Note that cost of RRU sharing is the performance loss caused by co-channel interference, which becomes severer as the cardinality of the consider coalition $\mathcal{T}_i$ increases. Therefore, the grand coalition formation, which implies that all the content objects share a common RRU, is not the best choice in most instances, and the cost part in \eqref{eqn:eff cap rrb} can control the scale of each coalition efficiently.
\begin{algorithm}[t!]
\caption{(A suboptimal coalitional game-based algorithm)} \label{alg:sub opt}
\begin{algorithmic}
 \STATE \textit{Step 1. RRU allocation}
\begin{itemize}
 \item \textbf{Initialization}:
 All content objects in $\mathcal{C}_T$ are divided into $N$ disjoint coalitions $\mathcal{T}_{1}, \ldots , \mathcal{T}_{N}$;
 \item \textbf{Repeat}: For each coalition $\mathcal{T}_{k}$
 \begin{itemize}
 \item $S_j$ in $\mathcal{T}_{k}$ negotiates with other coalitions $\mathcal{T}_{m}$, $m \neq k$, and $\varphi_{j}(\mathcal{T}_k)$ and $\varphi_{j}(\mathcal{T}_m)$ can be obtained due to \eqref{eqn:shap val uti}, and the total utility values of each coalitions is based on \eqref{eqn:tot uti};
 \item If the obtained utility values satisfy the preference relation criteria given in \eqref{eqn:pref rela}, then $S_j$ transfers into a new coalition $\mathcal{T}_{m}$;
 \end{itemize}
 \item \textbf{Termination}: When the members of $\mathcal{T}_{1}, \ldots , \mathcal{T}_{N}$ do not change.
\end{itemize}
 \STATE \textit{Step 2. RRH association by following Algorithm \ref{alg:rrh}.}
\STATE \textit{Step 3. Identify RRHs that are not required by any user, and put them into sleep mode.}
\end{algorithmic}
\end{algorithm}

\subsection{A Suboptimal RRU Allocation Algorithm}
Based on \eqref{eqn:shap val uti}, the payoff function of each content object is related to the members in its own coalition only, and thus can be solved by a hedonic method as well. To establish the preference relation criteria given in \eqref{eqn:pref rela}, the total utility of $\mathcal{T}_i$ can be written as
\begin{eqnarray}\label{eqn:tot uti}
u(\mathcal{T}_i) = \sum_{S_{j} \in \mathcal{T}_i} \varphi_{j}(\mathcal{T}_i).
\end{eqnarray}
Then a hedonic RRU allocation algorithm can be given, which is similar to Algorithm \ref{alg:rrh}, and the Nash stability of its partition results can be guaranteed by Proposition 5. Algorithm \ref{alg:sub opt} provides a sub-optimal solution with lower computational complexity, in which the joint design of RRH and RRU allocations are decoupled as two independent coalition formation games.

{In Step 1 of Algorithm 2 and Step 1-2 of Algorithm 3, we assumed that all the RRHs are active, and thus all of them should transmit at the same time. Moreover, each user accesses its nearest RRH that serves its desired content. These assumptions are consistent with those in Section II, which do not change the distribution of interference. Therefore, the utility values of coalition formation games can be obtained based on the analytical results. The turning-off process is just an additional step to avoid spending power on keeping unnecessary RRHs active, which is determined by the partition result only.}
{\subsection{Comparison of Computational Complexity of Algorithm 2 and Algorithm 3}
\subsubsection{Computational Complexity Analysis of Algorithm 2}
In Algorithm 2, both RRH and RRU are allocated by solving coalition formation games. The computational complexity of the proposed algorithm is mainly determined by the communication between different coalitions and the computation of the utility function. Denoting the number of RRHs in $\mathcal{C}_T$ as $D$, the obtained partition results of the RRH association in Algorithm 2 are maximum coalitions in \cite{b50}, and its computational complexity can be estimated as $O((D+M) 2^{2D + M})$. Then we analyze the computation complexity of RRU allocation, which is based on the pairwise negotiation process. As introduced in \cite{b54}, the RRU allocation can be accomplished in $O(L^2)$ steps of negotiation. In each negotiation operation, a coalition formation game-based RRH association problem should be solved due to the nested structure. Note that $M$ increases linearly as $L$ increases, e.g., $M = aL$, $0 < a < 1$. Therefore, the computational complexity of Algorithm 2 can be estimated as $O(L^2(D+aL) 2^{2D + aL})$.

\subsubsection{Computational Complexity Analysis of Algorithm 3}
Although the computation of the Shapley value is NP-Hard, it can be approximated efficiently and accurately, such as the linear approximation method in \cite{b60}, where the computational complexity can be estimated as $O(D^3)$ in our studied problem. Similar to RRH association in Algorithm 2, $N$ RRUs can be allocated among $L$ content objects by solving a hedonic coalition formation game, and its computational complexity can be estimated as $O((1+b)L 2^{(2+b)L})$, where $N = bL$, $0 < b < 1$. Since the computation of the Shapley value and the allocations of RRHs and RRUs can be executed independently, the computational complexity of Algorithm 3 can be estimated as $O(D^3 + (1+b)L 2^{(2+b)L} + (D+aL) 2^{2D + aL})$, which is lower than that of Algorithm 2.

The joint design of RRU allocation and RRH association in Sections IV and V is related to the results of the previous section, since the utility functions are formulated based on the derived effective capacity provided by Theorem 3. Moreover, by using effective capacity as a metric, the resource allocation results are jointly determined by the channel capacity and the QoS. In particular, the RRH association strategy can be taken as an example. In conventional allocation schemes, each RRH should serve the content requested by its nearest user on the criterion of throughput. However, the situation is different in our studied caching-based model. When its nearest user requests content that is not stored locally, a RRH may choose to serve other content in its local cluster cache instead, which can achieve higher effective capacity due to the improvement in QoS. There exist similar cases in RRU allocation. Therefore, the employment of caches has significant impact on resource allocation in our scenario. Moreover, our proposed resource allocation algorithms are also applicable for the non-cacheable content, which can improve the delay experience.}
\begin{figure}[t]\centering
\epsfig{file=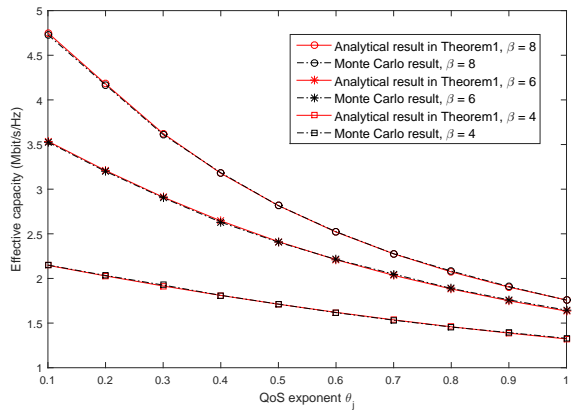, width=3.5in, clip=} \caption{Effective capacity of a typical user ($d_m$ = 50 m).}\label{fig:effcap_them1}
\end{figure}
\begin{figure} \centering
\subfigure[Effective capacity vs. size of cluster content cache]{\label{fig:subfig:aver_typa} \includegraphics[width=3.5in]{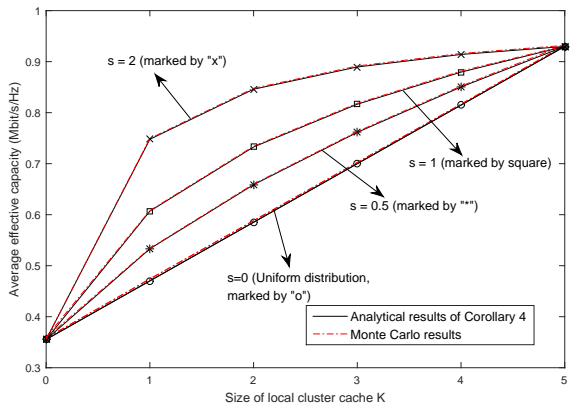}}
\subfigure[Average effective capacity to average power consumption ratio vs. size of cluster content cache]{\label{fig:subfig:aver_typb} \includegraphics[width=3.5in]{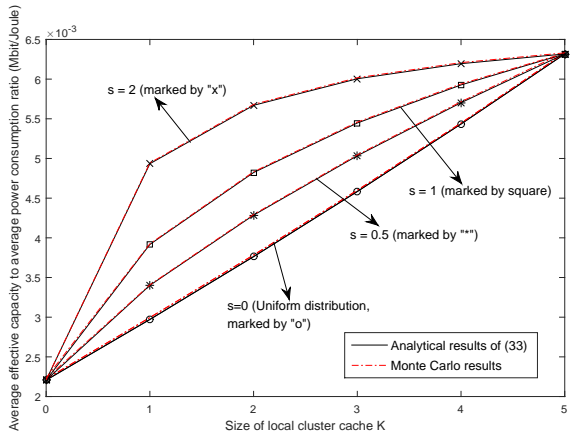}}
\caption{Performance of a typical C-RAN cluster.} \label{fig:aver_typ}
\end{figure}
\section{Simulation Results}
In this section, simulation results are provided to evaluate the performance of the cluster content caching structure in C-RANs. In particular, the channel are assumed to be block fading in each RRU, where $T = 1$ ms and $W = 1$ kHz. The densities of RRHs and users are given as $\lambda_{\mathrm{R}} = \lambda_{\mathrm{U}} = 5 \times 10^{-6}$. The cardinality of the required content set is $L = 5$, the size of each content object is set as $B_S = 1$ Mbits, and the popularity of content objects follows Zipf's law \cite{b33}. The power consumptions of RRHs in the active and the sleep modes are set as $P_{\mathrm{R}}^{\mathrm{act}} = 104$ W and $P^{\mathrm{sle}}_{\mathrm{R}} = 56$ W, respectively \cite{b32}, and the power consumption of cluster content caching and backhaul transmissions defined in are set as $P_{\mathrm{CC}} = 0.15$ W and $P_{\mathrm{BH}} = 10$ W.

{To verify the accuracy of our theoretical analysis, the effective capacity of a typical user is plotted in Fig. \ref{fig:effcap_them1}, where the path loss exponent is set as $\beta = 4, 6, 8$. As shown in the figure, the numerical results based on the analytical results given in Theorem 3 match the Monte Carlo results perfectly, which shows the validity of our theoretical derivations. Moreover, as the value of the path loss exponent $\beta$ increases, the effective capacity of a typical user increases. The reason is that the receive SINR can be improved as the signal to interference path loss ratio increases, e.g., $\gamma \sim d_m^{- \beta} / (\sum_{j \neq m} d_j^{- \beta}) = 1 / (\sum_{j \neq m} (d_m / d_j)^{\beta})$. Note that the serving RRH is usually nearer to the user than the interfering RRHs, e.g.,  $d_m \leqslant d_j$, and thus $\gamma$ is an increasing function of the path loss exponent $\beta$, and so is the effective capacity.

The performance of a typical cluster is evaluated in Fig. \ref{fig:aver_typ}, where the Zipf exponent is set as $s = 0, 0,5, 1, 2$, respectively, and the QoS exponents are set as $\theta_l^{\mathrm{T}} = 0.1$ and $\theta_l^{\mathrm{C}} = 0.6$, respectively. As shown in the figure, both the effective capacity and the energy efficiency increase as the size of cluster cache is enlarged, since more requests can be responded to locally with short delay. In particular, compared with the no-caching structure, e.g., $K = 0$, the average effective capacity and the average effective capacity to average power consumption ratio can be improved up to 0.57 Mbit/s/Hz and 0.004 Mbit/Joule when $K = 5$. Moreover, the performance of cluster content caching improves faster as $s$ increases, since the user interests converge to fewer popular content objects stored in the cluster cache.}

To further evaluate the performance improvement of cluster content caching, the performance of the proposed RRU allocation and RRH association algorithms is provided in Fig. \ref{fig:res_all}, where the effective capacity and the effective capacity to power consumption ratio are plotted, respectively. For the $i$-th RRU, the average effective capacity to power consumption ratio is defined as \eqref{eqn:ec pow ratio} on the following page.
\begin{figure*}[ht]
\begin{equation}\label{eqn:ec pow ratio}
\eta_{i} = \frac{\sum_{S_l}\bar{E}(\mathcal{R}_{S_l})}{\sum_{S_l} \mathcal{O}(\mathcal{R}_{S_l}) P_{\mathrm{R}}^{\mathrm{act}} + (N_T - \sum_{S_l} \mathcal{O}(\mathcal{R}_{S_l})) P^{\mathrm{sle}}_{\mathrm{R}} + \mathcal{O}(\mathcal{U}_T) P_{\mathrm{CC}} + \mathcal{O}(\mathcal{U}_C) P_{\mathrm{BH}}},~S_l \in \mathcal{T}_i.
\end{equation}
\hrulefill
\end{figure*}
As shown in Fig. \ref{fig:res_all}, both the effective capacity and the energy efficiency increase as $M$ increases and $\theta_l^\mathrm{C}$ decreases. Compared with the simulation results in Fig. \ref{fig:aver_typ}, the performance gains of cluster content caching can be enlarged to 0.95 Mbit/s/Hz and 0.0055 Mbit/Joule when $K=5$.

\begin{figure} \centering
\subfigure[Average effective capacity performance.]{\label{fig:subfig:effcap_hit} \includegraphics[width=3.5in]{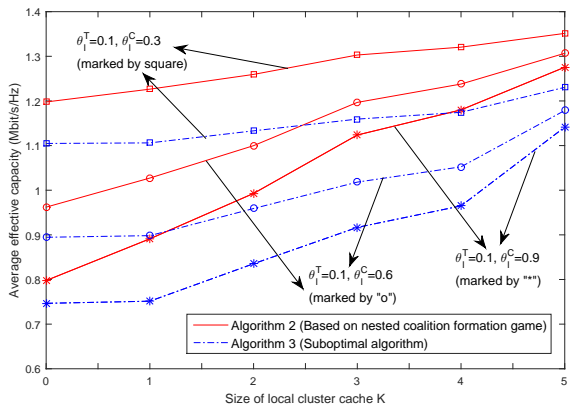}}
\subfigure[Average effective capacity to power consumption ratio]{\label{fig:subfig:en_hit} \includegraphics[width=3.5in]{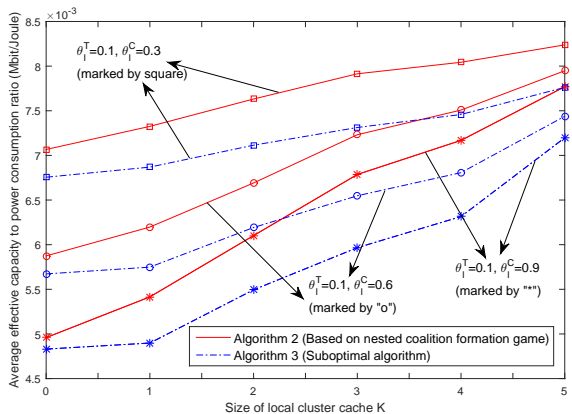}}
\caption{Performance of the proposed RRU allocation and RRH association algorithms.} \label{fig:res_all}
\end{figure}
\begin{figure} \centering
\subfigure[Average effective capacity performance.]{\label{fig:subfig:aver_typa} \includegraphics[width=3.5in]{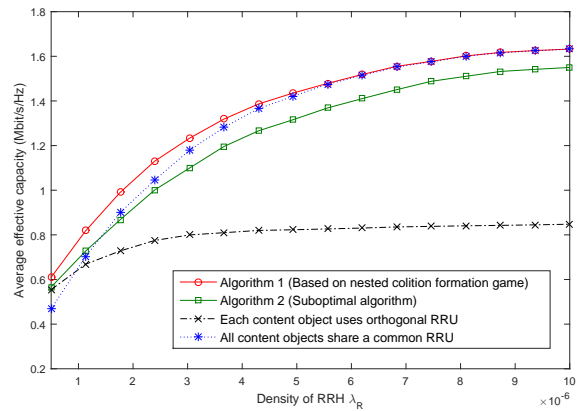}}
\subfigure[Average effective capacity to power consumption ratio.]{\label{fig:subfig:aver_typb} \includegraphics[width=3.5in]{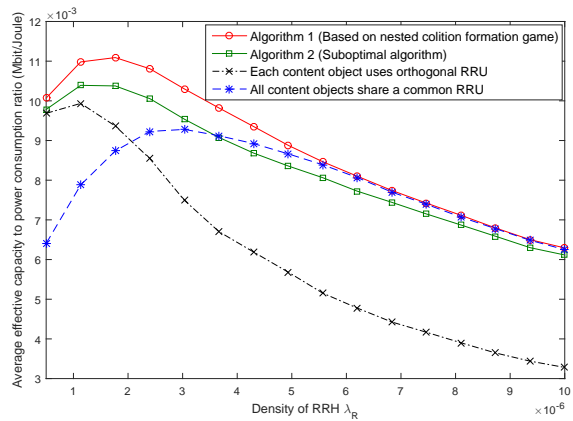}}\caption{Performance Comparison of different resource allocation schemes.} \label{fig:en_eff_den}
\end{figure}
In Fig. \ref{fig:en_eff_den}, the performance comparison of different RRU allocation and RRH association schemes is provided. Two previous resource allocation schemes are selected as two comparable schemes, which are the orthogonal RRU allocation scheme and the RRU full-reusing scheme. As shown in Fig. \ref{fig:subfig:effcap_hit}, Algorithm 2 can always achieve the best average effective capacity performance. There exists a performance gap between Algorithm 2 and Algorithm 3, since the utility function of RRU allocation in Algorithm 3 cannot ensure each content object makes the best choice. As $\lambda_{\mathrm{R}}$ increases, the interest conflict among content objects can be alleviated, and the performance of RRU full-reusing scheme approaches that of Algorithm 2. Although increasing the density of RRHs can always improve the effective capacity performance, it is not the best choice from an energy efficiency perspective. As shown in Fig. \ref{fig:subfig:en_hit}, the effective capacity to power consumption ratio increases in the low $\lambda_{\mathrm{R}}$ region, while it keeps decreasing in the high $\lambda_{\mathrm{R}}$ region. {Moreover, the case of multi-casting can be considered as a special case of our studied scheme, e.g., when content objects are served by using orthogonal RRUs, and the performance of the orthogonal RRU allocation scheme in Fig.6 can be treated as a lower bound on multi-casting.}

\section{Conclusion}
In this paper, a cluster content caching structure has been proposed in C-RANs, in which some requested content could be stored in local cluster content caches. By using a stochastic geometry-based network model, the effective capacity, which is defined as a link-level QoS metric, has been extended to our studied C-RAN scenario with content caching. Tractable expressions for the effective capacity and the energy efficiency have been derived to verify the performance gains of our proposed structure. The joint design of RRU allocation and RRH association has been studied to further improve the performance of cluster content caching, and two coalition formation game-based algorithms have been designed. The simulation results show that the effective capacity and the energy efficiency can be improved up to 0.57 Mbit/s/Hz and 0.004 Mbit/Joule when the number of required content objects is five, primarily due to relief of loading on the backhaul of the proposed structure. By employing the proposed optimization algorithms, the performance gains can be increased to 0.95 Mbit/s/Hz and 0.0055 Mbit/Joule, respectively.

\end{document}